%% file: main.tex
\makeatletter
\def\input@path{{preamble/}}
\makeatother
\documentclass[conference]{IEEEtran}
\IEEEoverridecommandlockouts

\RequirePackage{silence}
\usepackage[utf8]{inputenc} 
\usepackage[T1]{fontenc}    
\usepackage{microtype}      
\usepackage{booktabs,chemformula}
\usepackage{amsmath,amssymb,amsfonts}
\usepackage{amsthm}
\usepackage{graphicx}
\usepackage{colortbl}
\usepackage{hyperref}       
\usepackage{pifont}
\usepackage{cleveref}
\hypersetup{
    colorlinks,
    linkcolor={red!50!black},
    citecolor={blue!50!black},
    urlcolor={blue!80!black}
}

\usepackage{cite}
\usepackage{xcolor}
\usepackage{subcaption} 
\usepackage{url}            
\usepackage{nicefrac}       
\usepackage{xspace}
\usepackage{xifthen}
\usepackage[probability,operators]{cryptocode}
\usepackage{thmtools,thm-restate}	
\usepackage{tabularx}

\newcommand{\objectspace}{\ensuremath{\mathbb{O}}\xspace}
\newcommand{\secretspace}{\ensuremath{\mathbb{S}}\xspace}

\newcommand{\channel}{\ensuremath{C}\xspace}
\newcommand{\Ch}{\channel}

\newcommand{\Guesses}{\mathcal{W}}

\newcommand{\distset}[1]{\mathcal{D}(#1)}

\newcommand{\lnorm}[1]{\ensuremath{L_{{#1}}}\xspace}  
\newcommand{\loss}{\ensuremath{\ell}\xspace}

\newcommand{\Diam}[1]{\mathrm{diam}(#1)}
\newcommand{\MaxDist}[2]{d_{\max}(#1,#2)}
\newcommand{\ConvHull}[1]{\mathrm{ch}(#1)}
\newcommand{\Reals}{\mathbb{R}}
\newcommand{\Wide}[1]{~#1~}
\newcommand{\Ind}[1]{\mathbf{1}_{#1}}
\newcommand{\feasible}{\mathcal{F}}

\newcommand{\meleakage}{\ensuremath{\mathcal{L}}\xspace}

\newcommand{\mulcapacity}{\ensuremath{\mathcal{M\!L}}\xspace}

\newcommand{\Exp}{\textrm{E}}
\newcommand{\GEx}{{g_{\mathrm{x}}}}
\newcommand{\LEx}{{\loss_{\mathrm{x}}}}
\newcommand{\GS}[1][s]{{g_{#1}}}
\newcommand{\LS}[1][s]{{\loss_{#1}}}
\newcommand{\GP}{{g_P}}
\newcommand{\LP}{{\loss_P}}
\newcommand{\LorG}{{h}}
\newcommand{\Uni}{{\mathbf{u}}}
\newcommand{\WtoS}{Q^\pi}
\newcommand{\OptQ}[1][]{%
  \mathchoice
    {q_{#1}^\star}
    {q_{#1}^\star}
    {q_{#1}^{\kern-.6pt\raisebox{-1.4pt}{$\scriptstyle\star$}}}
    {q_{#1}^{\kern-.6pt\raisebox{-1.4pt}{$\scriptscriptstyle\star$}}}
}

\newtheorem{lemma}{Lemma}

\newtheorem{theorem}{Theorem}
\newtheorem{proposition}{Proposition}

\theoremstyle{definition}

\newcommand\eatspace[1]{}
\def\sectionautorefname{\S\kern.03em\eatspace}  
\def\subsectionautorefname{\S\kern.03em\eatspace}

\usepackage{cite}
\usepackage{amsmath,amssymb,amsfonts}
\usepackage{algorithmic}
\usepackage{graphicx}
\usepackage{textcomp}
\usepackage{xcolor}

\def\BibTeX{{\rm B\kern-.05em{\sc i\kern-.025em b}\kern-.08em
    T\kern-.1667em\lower.7ex\hbox{E}\kern-.125emX}}

\begin{document}

\title{Self-Defense: Optimal QIF Solutions and Application to Website Fingerprinting
\thanks{
The work of Andreas Athanasiou and Konstantinos Chatzikokolakis was supported by the project CRYPTECS, funded by the ANR (project number ANR-20-CYAL-0006) and by the BMBF (project number 16KIS1439).
The work of Catuscia Palamidessi was supported by the project HYPATIA, funded by the ERC (grant agreement number 835294). 
}
}

\author{
\IEEEauthorblockN{Andreas Athanasiou}
\IEEEauthorblockA{\textit{INRIA, École Polytechnique} \\
Palaiseau, France \\
andreas.athanasiou@inria.fr}
\and
\IEEEauthorblockN{Konstantinos Chatzikokolakis}
\IEEEauthorblockA{
\textit{National and Kapodistrian University of Athens}\\
Athens, Greece \\
kostasc@di.uoa.gr}
\and
\IEEEauthorblockN{Catuscia Palamidessi}
\IEEEauthorblockA{\textit{INRIA, École Polytechnique}    \\
Palaiseau, France \\
catuscia@lix.polytechnique.fr}
}

\maketitle

\begin{abstract}

Quantitative Information Flow (QIF) provides a robust information-theoretical framework for designing secure systems with minimal information leakage. While previous research has addressed the design of such systems under hard constraints (e.g. application limitations) and soft constraints (e.g. utility), scenarios often arise where the core system's behavior is considered fixed. In such cases, the challenge is to design a new component for the existing system that minimizes leakage without altering the original system. 

In this work we address this problem by proposing optimal solutions for constructing a new row, in a known and unmodifiable information-theoretic channel, aiming at minimizing the leakage. We first model two types of adversaries: an \emph{exact-guessing} adversary, aiming to guess the secret in one try, and a \emph{s-distinguishing} one, which tries to distinguish the secret $s$ from all the other secrets.
Then, we discuss design strategies for both fixed and unknown priors by offering, for each adversary, an optimal solution under linear constraints, using Linear Programming.

We apply our approach to the problem of website fingerprinting defense, considering a scenario where a site administrator can modify their own site but not others. We experimentally evaluate our proposed solutions against other natural approaches. First, we sample real-world news websites and then, for both adversaries, we demonstrate that the proposed solutions are effective in achieving the least leakage. 
Finally, we simulate an actual attack by training an ML classifier for the s-distinguishing adversary and show that our approach decreases the accuracy of the attacker.

\end{abstract}


\input{1_introduction}
\input{2_preliminaries}
\input{3_predicates}

\input{4_fixed-prior}
\input{5_unknown-prior}

\input{6_seb}



\input{8_evalution}

\input{discussion}

\bibliography{references}
\bibliographystyle{IEEEtran}

\appendix
\input{9_appendix}

\end{document}

%% file: 1_introduction.tex
\section{Introduction}

A fundamental link in the chain of secure system designing is the quantification of the protection its users enjoy, or conversely, determining how much information is \emph{leaked} to an adversary. A system is traditionally described by a channel $C$, which is a probability distribution over its observable outputs. The adversary is assumed to have some \emph{prior} knowledge before the system's execution, and the goal is to determine how much information the adversary has afterwards, due to the system's leakage during execution. 

Quantitative Information Flow (QIF) \cite{qif_book} studies measures of information \emph{leakage}. The main idea is simple: it quantifies a system's leakage by \emph{comparing} the secret's \emph{vulnerability} or \emph{risk} before (prior) and after (posterior) executing the system. 

When designing a system from scratch, one can use QIF to create a channel that minimizes leakage.
However here we consider a different scenario where a system and its channel are already given, and we just want to add a new component. We assume that we can design this new component in any way we want (or perhaps under some constraints, which we discuss later on), but we cannot modify anything in the existing system.

For example consider the following problem: the administrator of a website $s$ wants to
prevent Website Fingerprinting (WF); that is, prevent an adversary from inferring
which site a user visits, over an encrypted connection, from observable
information such as the page size. To achieve this goal, the administrator
wants to make his website similar to a set $\secretspace\setminus\{s\}$ of other websites.
The behavior of those websites is completely
known: for each $t\in\secretspace\setminus\{s\}$, the distribution $\Ch_t$ of
observations produced by a visitor of $t$ is given.
Note that $t$ may or
may not be using obfuscation to prevent fingerprinting; this is reflected
in the distribution $\Ch_t$. The administrator of $s$ cannot possibly
modify the \emph{other} websites, but he has full control over \emph{his own}. 
For instance, he can pad
his site to make the page size match any other known
page. Moreover,
this can be done probabilistically, by constructing an arbitrary distribution $q$
of observations to be produced by the visitors of $s$.

We call this scenario \emph{self-defense}, since the administrator of the website $s$
tries to protect his site using his own means. 
The fundamental question then is:
\emph{which distribution $q$ provides the best self-defense}?

Note that, in this work, we assume that the administrator cannot rely on the cooperation of any other websites. This is a crucial assumption, that significantly affects the optimal solution.
Indeed, it is well-known in game theory that cooperative and non-cooperative games have very different outcomes. For instance,  employing a common defense strategy can benefit all parties involved, while if only one site employs that particular strategy and the others do not, that site may become more vulnerable. Take, for instance, padding: if all administrators pad their website to the same value $v$, they will generate the same observation. However, if only one website implements the padding to $v$ while the others do not, the attacker will single out that website more easily.

To state the problem formally in the language of QIF, consider
a system described by the channel $\Ch : \secretspace\to \distset\objectspace$.
All rows $\Ch_t, t\neq s$ are fixed,
but we can freely choose the row $\Ch_s$.
Concretely, we want to construct a new
channel $\Ch^q : \secretspace\to \distset \objectspace$, such that $\Ch^q_{t} = \Ch_{t}$ for all $t \neq s$, while the
row $\Ch^q_s$ can be set to a new distribution 
$ q \in \distset\objectspace$. The question now becomes:
\emph{how to choose $q$ to minimize the leakage of $C$}?

Of course this problem is not specific to website fingerprinting; it arises in any
QIF scenario in which
each secret corresponds to a different \emph{user} or \emph{entity},
which has full control over his own part of the system, but no control over the remaining system. 
Note that the knowledge of $\Ch_t$ is vital for choosing $q$; we want to make
our secret as similar as possible to the other ones. 

One may think of choosing $q$ so to replicate a specific site $t'$, or to emulate the ``average'' website. 
These naive approaches, however,  fall short of providing an optimal solution. In this work, we tackle this problem by using QIF, rather than resorting to an ad-hoc solution.
In particular, we will use the version of QIF known as $g$-leakage~\cite{alvim2012measuring}, which allows us to express the capabilities and goals of a large class of adversaries. 
We consider two kinds of adversaries: one that tries to exactly guess the secret in one try (\emph{exact-guessing}) and one that tries to distinguish the secret $s$ from all the other secrets (\emph{s-distinguishing}). 

Until now we have assumed complete freedom in designing the new row. In some situations, however, such freedom could be unrealistic: there may be technical constraints, which of course will increase the complexity of the problem. For instance, in the WF example, one can easily increase the page size by padding but cannot reduce it (at least, not to an arbitrary value, as infinite compression is not possible). Therefore, we also consider the problem of finding the optimal solution in the presence of constraints on the new row $\Ch_s$.

We summarize our contributions as follows:

\begin{itemize}
    \item For a fixed prior, we show how to optimally create $q$ for both the exact-guessing and $s$-distinguishing adversary (\Cref{sec:fixed-prior}).
    \item When the prior is unknown, we first discuss how to optimally create $q$ for the exact-guessing adversary, and then explore the $s$-distinguishing adversary (\Cref{sec:unknown-prior}). We demonstrate that the latter reduces to the problem of the Smallest Enclosing Ball (\Cref{sec:seb}), offering an optimal solution and discussing computationally lighter alternatives.
    \item In all cases we demonstrate how to incorporate the constraints in the proposed solutions using Linear Programming (LP). 
    \item Finally, we evaluate our proposed solutions by simulating a website fingerprinting attack and compare them with other conventional approaches for both adversaries (\Cref{sec:exp}). The experiments confirm that our approaches offer indeed the minimum leakage and decrease the accuracy of the attack. 
\end{itemize}

\subsection{Related Work}

The works by Simon et al. \cite{simon2022minimizing} and Reed et al. \cite{reed2021optimally} were motivated by a similar traffic analysis attack; they explored how a server can efficiently pad its files in order to minimize leakage. 
In their setting, a user selects a file from the server which is padded according to some probabilistic mechanism.
Meanwhile, an adversary, observing the encrypted network, attempts to guess the selected file. Our case differs as we strive to make a website indistinguishable from others (which are considered fixed), rather than making the pages of the website indistinguishable among themselves.

On the other hand, designing a channel to minimize leakage by being able to manipulate all it's rows (instead of only one, as in our case) under some constraints, has been studied in the literature \cite{KhouzaniM17, e20090675}.
In fact, by properly selecting these constraints one can fix the value of all rows but one, essentially mimicking our scenario. Using the techniques from
\cite{KhouzaniM17} 
under such constraints, one can arrive at optimization programs similar to those of \autoref{optimazation_prop2} and \ref{prop:unknown-exact-lp}. However, studying in depth the specific problem
of single row optimization is interesting since,
first, we can give direct solutions such as
\autoref{prop:unknown-exact-optimal} and
second, we can give capacity solutions for an
s-distinguishing adversary, which is outside
the scope of \cite{KhouzaniM17} and a main contribution of our work.

P. Malacaria et al. \cite{malacaria_deterministic} discussed the complexity of creating a deterministic channel aimed at minimizing leakage for a given leakage metric and a specific prior while satisfying a set of constraints. They first showed that this is an NP-hard problem and then they proposed a solution for a particular class of constraints, by introducing a greedy algorithm. The authors show that this algorithm offers optimal leakage across most entropy measures used in the literature.

Moreover, one could opt to explore the problem within the popular  framework of Differential Privacy (DP) \cite{dppaper1}. DP aims to protect the privacy of an individual within a statistical database when queried by an analyst (considered the adversary) seeking aggregated information. A trusted central  server aggregates users' data and introduces noise before publishing the noisy result to the analyst, making it hard for her to distinguish an individual's value. 
The (central) DP solution however is not suitable in our case, because there is no central server, and the attacker can observe directly the obfuscated version of the secret (the output of the channel). 

In contrast, Local Differential Privacy (LDP) \cite{whattolearn_privately} removes the need for a central entity. Instead, users autonomously inject noise into their data. Notably, all users employ the same probabilistic mechanism to perturb their data without knowledge of other users' values and perturbations. 
LDP is more suitable for the situation we are considering. However, our scenario involves 
only one secret having the ability to be modified, and assumes knowledge of how the other secrets are treated by the system.
Furthermore, LDP (as well as DP) is 
a worst-case metric that, for every observable, requires a certain level of indistinguishability between our secret $s$ (corresponding to the row that we want to modify) and every other secret (row) $t$. This may be impossible to 
satisfy, as it would require that all the 
rows are almost indistinguishable as well, but
as we already explained, the other rows are already given and cannot be modified.

Another metric proposed to measure the vulnerability of a secret is the multiplicative  Bayes risk leakage $\beta$
\cite{Cherubin17POPETS}, which was inspired by the cryptographic notion of \emph{advantage}. $\beta$
corresponds to the (multiplicative) $g$-leakage in the case of a  Bayesian attacker. 
In \cite{bayes_security}, the authors showed that $\beta$ can also be used to quantify the risk for the two most vulnerable secrets. 

The related problem of \emph{location privacy} has also been studied in the QIF literature using game theory and LP approaches. For instance, \cite{shokri2012protecting} discussed the optimal user strategy when a desired Quality of Service (QoS) and a target level of privacy are given, taking into account the prior knowledge of the adversary. Moreover, \cite{shokri2015privacy}, motivated by DP and distortion-privacy (i.e. inference error), provided a utility-maximizing obfuscation mechanism with formal privacy guarantees, aiming to solve the trade-off between QoS and privacy. Their mechanism is based on formulating a Stackelberg game and solving it via LP, using the DP guarantee as a constraint. They claim their approach is utility-wise superior to DP while offering the same privacy guarantees. Furthermore, \cite{bordenabe2014optimal} studied how to maximize the QoS while achieving a certain level of geo-indistinguishability using LP, and discussed methods to reduce the constraints from cubic to quadratic, significantly reducing the required computation time.
Note that an extension of this problem, which aims to protect entire location trajectories rather than individual locations (known as \emph{trajectory privacy}), has also been studied using similar techniques \cite{lbs_trajectory}.

%% file: 2_preliminaries.tex
\section{Preliminaries} \label{prelim}

The basic building blocks of QIF are summarized in the following sections.

\subsubsection*{Prior vulnerability and risk}

A natural framework for expressing vulnerability is in terms of \emph{gain}
functions.
Let $\secretspace$ be the set of all possible secrets (e.g. all
valid passwords); the adversary does not fully know the secret, but he possesses some
\emph{probabilistic knowledge} about it, expressed as a \emph{prior
distribution} $\pi:\distset\secretspace$. In order to exploit this knowledge,
the adversary will perform some \emph{action} (e.g. make a guess about the
password in order to access the user's system); the set of all available
actions is denoted by $\Guesses$. Often we take $\Guesses = \secretspace$,
meaning that the adversary is trying to guess the exact secret; however
expressing certain adversaries might require a different choice of
$\Guesses$ (see \autoref{sec:guessing-predicates} for an example).

A gain
function $g(w,s)$ expresses the gain
obtained by performing the action $w\in\Guesses$ when the secret is $s\in\secretspace$.
The \emph{expected gain} (wrt $\pi$) of action $w$ is $\sum_s \pi_s g(w,s)$.
Being rational, the adversary will choose the action that maximizes his expected gain;
we naturally define
\emph{$g$-vulnerability} $V_g(\pi)$ as the expected gain
of the \emph{optimal} action:
\begin{align*}
	V_g(\pi)	&\Wide= \textstyle\max_w \sum_s \pi_s g(w,s)
	~.
\end{align*}

Alternatively, it is often convenient
to measure the adversary's \emph{failure} (instead of his success).
This can be done in a dual manner by expressing the \emph{loss} $\ell(w,s)$ occurred by the action $w$.
The adversary tries to minimize his loss, hence \emph{$\loss$-risk}
is defined as the expected loss of the optimal $w$.
\begin{align*}
	R_\loss(\pi)	&\Wide= \textstyle\min_w \sum_s \pi_s \loss(w,s)
	~.
\end{align*}
Note that $V_g(\pi)$ is a \emph{vulnerability} function,
expressing the adversary's
\emph{success} in achieving his goal,
while $R_\loss(\pi)$ is a \emph{risk} function (also known as
\emph{uncertainty} or \emph{entropy}), expressing the adversary's failure.

\subsubsection*{Posterior Vulnerability and Leakage}

So far we modeled the adversary's success or failure \emph{a-priori}, that is given
only some prior information $\pi$, before executing the system of interest.
Given an output $o$ of a system $\Ch:\secretspace\to\objectspace$, the adversary
applies Bayes law to convert his prior $\pi$ into a \emph{posterior} knowledge
$\delta$, which is exactly what causes information leakage.
Hence, $V_g(\delta)$ expresses the adversary's success \emph{after observing $o$}
(and similarly for $R_\ell(\delta)$). Since outputs are selected randomly, and each
produces a different posterior, we can intuitively define the
\emph{posterior} $g$-vulnerability $V_g(\pi,\Ch)$ and the posterior $\ell$-risk
$R_\ell(\pi,\Ch)$,\footnote{$R_\ell(\pi,\Ch)$ in
often called \emph{Bayes risk} (wrt $\ell$).} as the expected value
of $V_g(\delta)$ and $R_\ell(\delta)$ respectively:
\begin{align*}
	V_g(\pi,\Ch)
		&\Wide= \Exp [V_g(\delta)]
		\textstyle
		\Wide= \sum_o \max_w \sum_s \pi_s \Ch_{s,o} g(w,s)
	~, \\
	R_\loss(\pi,\Ch)
		&\Wide= \Exp [R_\ell(\delta)]
		\textstyle
		\Wide= \sum_o \min_w \sum_s \pi_s \Ch_{s,o} \loss(w,s)
	~.
\end{align*}

An adversary might succeed in his goal -- say, to guess the user's password -- due to two
very distinct reasons: either because the user is choosing weak passwords (that is, the
\emph{prior} vulnerability is high), or because the \emph{system} is leaking information
about the password, causing the \emph{posterior} vulnerability to increase. When
studying the privacy of a system we want to focus on the information leak caused by the system \emph{itself}.
This is captured by the notion of \emph{leakage}, the fundamental quantity in QIF, which simply
compares the vulnerability before and after executing the system.
The \emph{multiplicative} leakage of $\Ch$ wrt $\pi$ and $g/\ell$ is defined as:
\[
	\textstyle
	\meleakage_g(\pi,\Ch) \Wide= \frac{V_g(\pi,\Ch)}{V_g(\pi)}
	~,
	\qquad
	\meleakage_\ell(\pi,\Ch) \Wide= \frac{R_\ell(\pi)}{R_\ell(\pi,\Ch)}
	~.
\]
A leakage of $1$ means \emph{no leakage} at all: the prior and posterior
vulnerability (or risk) are exactly the same. On the other hand, a \emph{high
leakage} means that, after observing the output of the channel, the adversary
is more successful in achieving his goal than he was before. Note that $V_g$
always \emph{increases} as a result of executing the system, while $R_\ell$
always \emph{decreases}, this is why the fractions in the definitions of
$\meleakage_g$ and $\meleakage_\ell$ are reversed.%
\footnote{In \cite{ChatzikokolakisCPT23}, Bayes \emph{security} is defined as
$\beta(\pi,C) = 1/\meleakage_\ell(\pi,\Ch)$; in this paper we use leakage, which is
the standard notion in the QIF literature; all results can be trivially translated to $\beta(\pi,C)$.}
We use $\LorG$ to denote a function that can be either a gain or a loss function,
which is useful to treat both types together

\subsubsection*{Guessing the exact secret}

The simplest and arguably most natural adversary is the one that tries to 
guess the \emph{exact secret} in one try.
We call this adversary \emph{exact-guessing};
he can be easily expressed by
setting $\Guesses = \secretspace$, and defining gain and loss as follows:
\[
	\GEx(w,x) \Wide= \Ind{\{w\}}(x)
	~,
	\quad
	\LEx(w,x) \Wide= 1-\Ind{\{w\}}(x)
	~,
\]
where $\Ind{S}(x)$ is the indicator function (equal to $1$ iff $x\in S$ and $0$ otherwise).
For this adversary, vulnerability and risk reduce to the following expressions:
\begin{align*}
	V_\GEx(\pi) &= \max_s \pi_s
		~,
		&V_\GEx(\pi,\Ch) &= \textstyle\sum_y \max_s \pi_s \Ch_{s,o}
		~,\\
	R_\GEx(\pi) &= 1-V_\GEx(\pi)
		~,
		&R_\LEx(\pi,\Ch) &= \textstyle 1-V_\GEx(\pi,\Ch)
		~.
\end{align*}
$V_\GEx$ is known as \emph{Bayes vulnerability}, while
$R_\LEx$ as \emph{Bayes error} or \emph{Bayes risk}.\footnote{The term
Bayes risk has been used in the literature to describe both $R_\LEx$ and $R_\ell$ (for arbitrary $\ell$).}

\subsubsection*{Capacity}

Finally, the notion of \emph{$g/\ell$-capacity} is simply the worst-case leakage of a system
wrt all priors. Bounding the capacity of a system means that its leakage will be bounded
independently from the prior available to the adversary.
\[
	\textstyle
	\mulcapacity_\LorG(\Ch) \Wide= \max_\pi \meleakage_\LorG(\pi,\Ch)
	~,\quad
	h\in\{g,\loss\}
	~.
\]
Denote by $\Uni$ the uniform prior, and by  $\Uni^{s,t}$ the prior that is
uniform among these two secrets, i.e. it assigns
probability $\nicefrac{1}{2}$ to them.
For an exact-guessing adversary (i.e. for $\GEx,\LEx$),
the following results are known:

\begin{theorem}\label{thm:mbc}
	$\mulcapacity_\GEx(\Ch)$ is equal to:
	\[
		\textstyle
		\mulcapacity_\GEx(\Ch)
		\Wide=
		\meleakage_\GEx(\Uni,\Ch)
		\Wide=
		\sum_o \max_s \Ch{s,o}
		~.
	\]
\end{theorem}

In the following, we use the $\lnorm{1}$-metric
$\|q-q'\|_1$ as the distance between probability distributions
$q,q'$; note that this is equal to twice their \emph{total variation} distance,
since we only consider discrete distributions.
Denote by $\MaxDist{S}{T}$ the maximum distance between elements of sets $S,T$,
and by $\Diam{S} = \MaxDist{S}{S}$ the
\emph{diameter} of $S$.
Finally,
given a subset of secrets $P\subseteq\secretspace$,
we denote by
$\Ch_P\subset \distset\objectspace$ the set of rows of $\Ch$
indexed by $P$
(note that rows are probability distributions); the set of
all rows will be $\Ch_\secretspace$ while the $s$-th row
will be denoted by $\Ch_s$.


\begin{theorem}[\cite{ChatzikokolakisCPT23}]\label{thm:risk-mbc}
	$\mulcapacity_\LEx(\Ch)$ is equal to:
	\[
		\mulcapacity_\LEx(\Ch)
		\Wide=
		\meleakage_\LEx(\Uni^{s,t},\Ch)
		\Wide=
		\frac{1}{ 1 - \frac{1}{2}\Diam{\Ch_\secretspace} }
		~.
	\]
	where $s,t\in\secretspace$ realize $\Diam{\Ch_\secretspace}$.
\end{theorem}


%% file: 3_predicates.tex
\section{Guessing predicates and distinguishing a specific secret}\label{sec:guessing-predicates}

Although systems are commonly evaluated against the exact-guessing adversaries
$\GEx,\LEx$,
there are many practical scenarios in which 
the adversary is not necessarily interested in guessing the \emph{exact} secret.
A general class of such adversaries are those aiming at guessing only a \emph{predicate}
$P$
of the secret, for instance \emph{is the user located close to a hospital?} or
\emph{is the sender a male?}. We call this adversary \emph{$P$-guessing}.

A predicate can be described by a subset of secrets $P \subseteq
\secretspace$, those that satisfy the predicate; the rest are denoted by 
$\neg P = \secretspace\setminus P$. A $P$-guessing adversary can be easily captured
using gain/loss functions, by setting $\Guesses = \{P , \neg P \}$ and defining
\[
	\GP(w,s) \Wide= \Ind{w}(s)
	~,
	\quad
	\LP(w,s) \Wide= 1 - \Ind{w}(s)
	~,
\]
for $w \in \Guesses$. Then $\meleakage_{\GP}$ and $\meleakage_{\ell_P}$ measure the system's leakage
wrt an adversary trying to guess $P$.

Although this class of adversaries was discussed in
\cite{alvim2012measuring}, it has received little attention in the QIF literature.
In this section, we first show that $\meleakage_{\GP}$
can be expressed as the Bayes leakage $\meleakage_\GEx$ of a properly constructed pair of prior/channel
(and similarly for $\meleakage_{\LP}$). Then, we study the capacity problem for this adversary.

\subsubsection*{Expressing $\meleakage_{\GP}$ as $\meleakage_{\GEx}$}
Since every secret belongs to either $P$ or $\neg P$,
given a prior $\pi\in\distset\secretspace$ we can construct a \emph{joint distribution}
between secrets $s$ and classes $w$, where $\textrm{Pr}(w,s) = \pi_s \cdot \Ind{w}(s)$
gives the probability to have both $w$ and $s$ together. This joint distribution can be factored
into a \emph{marginal} distribution $\rho^\pi\in\distset\Guesses$,
and \emph{conditional} distributions -- i.e. a \emph{channel} --
$\WtoS:\Guesses\to\secretspace$:%
\footnote{In case $\rho^\pi_w = 0$ we can define the row $\WtoS_w$ arbitrarily.}
\begin{align}
	\rho^\pi_w
		&\Wide=	\textrm{Pr}(w)
		\kern16pt \Wide= \textstyle\sum_{s\in w} \pi_w
	~,\label{eq:rho}\\
	\WtoS_{w,s}
		&\Wide=	\textrm{Pr}(s \ |\ w)
		\Wide=
			\frac{\pi_s \cdot \Ind{w}(s)}{\rho^\pi_w}
		~.
		\label{eq:B}
\end{align}
In the notation we emphasize that both $\rho^\pi,\WtoS$ depend on $\pi$.%
\footnote{They also depend on $P$, which is clear from the context hence omitted for simplicity.}

This construction allows us to express $\GP$-leakage as
$\GEx$-leakage, for
a different pair of prior/channel, as stated in the following lemma.

\begin{lemma}\label{lem:p-vs-bayes}
	Let $\pi\in\secretspace, P \subseteq \secretspace$
	and define $\rho^\pi,\WtoS$ as in \eqref{eq:rho},\eqref{eq:B}.
	Then for all channels $C$ it holds that:
	\begin{align*}
		\meleakage_\GP(\pi, C)
			&\Wide= \meleakage_\GEx(\rho^\pi, \WtoS C)
		~, \\
		\meleakage_\LP(\pi, C)
			&\Wide= \meleakage_\LEx(\rho^\pi, \WtoS C)
		~.
	\end{align*}
\end{lemma}


\subsubsection*{Capacity}
\autoref{lem:p-vs-bayes} can be used to obtain capacity results for any $P$-guessing adversary.

\begin{restatable}{theorem}{ThmPGuessingCapacity}\label{thm:p-guessing-capacity}
	For any $\Ch$ and $P\subseteq\secretspace$ it holds that
	\begin{align*}
		\mulcapacity_{\GP}(\Ch)
		&\Wide=
		\meleakage_{\GP}(\Uni^{s,t},\Ch)
		\Wide=
		1 + \frac{1}{2}\; \MaxDist{\Ch_{P}}{\Ch_{\neg P}}
		~,
		\\
		\mulcapacity_{\LP}(\Ch)
		&\Wide=
		\meleakage_{\LP}(\Uni^{s,t},\Ch)
		\Wide=
		\frac{1}{ 1 - \frac{1}{2}\; \MaxDist{\Ch_{P}}{\Ch_{\neg P}} }
		~,
	\end{align*}
	for $s\in P,t\in\neg P$ realizing $\MaxDist{\Ch_{P}}{\Ch_{\neg P}}$.
\end{restatable}


\subsection{Distinguishing a specific secret}\label{sec:s-distinguishing}

As discussed, an adversary of particular interest is
one that tries to \emph{distinguish a secret $s$} from all other secrets,
that is to answer the question \emph{is the secret $s$ or not?}.
We call this adversary \emph{$s$-distinguishing}.

Clearly, this is a special case of a $P$-guessing adversary, for
$P = \{ s \}$. In this case, we simply write $s,\neg s$ instead of $P,\neg P$
respectively. Hence $g_s = \GP, \loss_s = \LP$ are the gain/loss functions modeling an
$s$-distinguishing adversary, and $\meleakage_{g_s}, \meleakage_{\loss_s}$ measure
the system's leakage wrt this adversary.

For this choice of $P$, the binary channel $\WtoS C:\{s, \neg
s\}\to\objectspace$ (see \eqref{eq:B}) has a clear interpretation. Its first
row is simply $\WtoS_s C = C_s$; more interestingly, its second row $\WtoS_{\neg s}C$ models
the behavior of an \emph{average} secret \emph{other than $s$}: choosing
randomly (wrt $\pi$) a secret $t \neq s$ and then executing $C$ on $t$
produces outputs distributed according to $\WtoS_{\neg s}C$. The usefulness of
this distribution will become apparent in the following sections.

Finally \autoref{thm:p-guessing-capacity} provides a direct solution of the capacity problem
for $s$-distinguishing adversaries.
The capacity will be given by the maximum distance $\MaxDist{C_s}{C_{\neg s}}$ between the
row $C_s$ and all other rows of the channel.

%% file: 4_fixed-prior.tex
\section{Optimizing $q$ for a fixed prior $\pi$}\label{sec:fixed-prior}


We return to the problem discussed in the introduction, that is choosing the distribution of
observations $q$ produced by our secret of interest $s$. The QIF theory gives us a clear
goal
for this choice: we want to find the distribution $q\in\feasible$ that \emph{minimizes the leakage}
of the channel $C^q$.

Here $\feasible\subseteq\distset\objectspace$ denotes the set of \emph{feasible} solutions
for $q$; this set can be arbitrary, and is typically determined by practical
aspects of each application (e.g. it might be possible to increase the page size
by padding, but not to decrease it). The only assumption we make in this paper is that $\feasible$
can be expressed in terms of \emph{linear inequalities}.
A solution $\OptQ$ will be called \emph{optimal} iff it is both feasible and minimizes
leakage among all feasible solutions.

We start with the case when we have a specific prior $\pi$
modeling the system's usage profile, and we want to minimize leakage wrt that prior.
Since the behavior $C_t$ of all secrets $t\neq s$ is considered fixed, and their relative
probabilities are dictated by our fixed $\pi$, we know exactly how an ``average'' secret other than
$s$ behaves: it produces
observables with probability $\WtoS_{\neg s}C$ (see \autoref{sec:s-distinguishing}).
Conventional wisdom dictates that
we should try to make $s$ similar to an average non-$s$ secret, that is 
choose 
\[
	\OptQ\Wide =\WtoS_{\neg s}C 
\]
as our output distribution. Note that this $\OptQ$ is a \emph{convex combination}
of $C$'s rows (other than $s$), the elements of $\WtoS_{\neg s}$ being the convex coefficients.

We first show that this intuitive choice is indeed meaningful in a precise sense:
it minimizes leakage%
\footnote{Note that, for fixed $\pi$, optimizing $\meleakage_{g},\meleakage_{\loss}$
is equivalent to optimizing $V_g,R_\loss$ since the prior vulnerability/risk is constant.}
wrt an \emph{$s$-distinguishing} adversary (\autoref{sec:s-distinguishing}).
In fact, for this adversary the resulting channel has no leakage at all;
the adversary is trying to distinguish $s$ from $\neg s$, but the two cases
produce identical observations.


\begin{proposition}\label{prop-conv-hull-optimal}
	Let $\pi\in\distset\secretspace$, $\Ch:\secretspace\to\objectspace$, $s\in\secretspace$
	and let $\OptQ = \WtoS_{\neg s}C$. The channel $C^{\OptQ}$ has no leakage
	wrt an $s$-distinguishing adversary, that is:
	\[
		\meleakage_{\LorG}(\pi,\channel^{\OptQ})
			\Wide=
		1
		~,
		\quad \LorG\in\{ \GS, \LS \}
		~.
	\]
	
\end{proposition}

An immediate consequence of \autoref{prop-conv-hull-optimal} is that,
if $\OptQ=\WtoS_{\neg s}C \in\feasible$, then it is optimal for an $s$-distinguishing adversary.
However, it could very well be the case that this construction is infeasible,
in which case finding the optimal feasible solution is more challenging.

Moreover, somewhat surprisingly, it turns out that even when  $\OptQ = \WtoS_{\neg s}C$ is
feasible, it does
\emph{not} minimize leakage wrt an exact-guessing adversary. Consider the following example:
\[
	\pi \Wide= (0.47, 0.29,   0.24)
	\qquad
	\Ch^q \Wide=
	\begin{bmatrix}
	q_1 & q_2 \\
	0.05  &  0.95  \\
	0.58  & 0.42
	\end{bmatrix}
\]
For this channel and prior we get
$\OptQ = \WtoS_{\neg s}C = (0.29, 0.71)$.
Consider also the distribution
$q = (0.42, 0.58)$. We have that
\begin{align*}
	\meleakage_{\LEx}(\pi, \channel^{\OptQ}) &\Wide\approx 1.1 ~,
		& \meleakage_{\LS}(\pi, \channel^{\OptQ}) &\Wide= 1 ~,
	\\
	\meleakage_{\LEx}(\pi, \channel^{q}) &\Wide= 1 ~,
		& \meleakage_{\LS}(\pi, \channel^{q}) &\Wide\approx 1.01 ~.
\end{align*}
We see that $\OptQ$ does indeed minimize leakage for an $s$-distinguishing adversary
(in fact, the leakage becomes $1$),
but it does not minimize the exact-guessing leakage; the latter is minimized by $q$.

Although the simple construction $\OptQ=\WtoS_{\neg s}C$ does not always produce
an optimal solution, as discussed above,
finding an optimal one is still possible for all adversaries and
arbitrary $\feasible$, via linear programming.

\begin{proposition} \label{optimazation_prop2}
	The optimization problem
	\begin{equation}
		\OptQ \Wide{:=} \argmin_{q\in\feasible} \meleakage_\LorG(\pi,\Ch^q)	
		~,\quad \LorG\in\{ \GEx, \GS, \LEx, \LS \} ~,
		\label{eq9324}
	\end{equation}
	can be solved in polynomial time via linear programming.
\end{proposition}

For instance, the optimal solution for the
$\GEx$ adversary is given by the following linear program
(recall that we assumed $q \in \feasible$ to be expressible in terms
of linear inequalities):

\begin{align*}
	\text{minimize} \qquad & {\textstyle\sum_{o\in\objectspace}} z_o
	\\
	\text{subject to}\qquad & q \in \feasible
	\\
	& z_o \ge \max_{t\in\secretspace\setminus\{s\}} \pi_t C_{t,o}    &\forall o\in \objectspace
	\\
	& z_o \ge \pi_s q_o  &\forall o\in \objectspace
\end{align*}
While for $\GS$, the program is similar:
\begin{align*}
	\text{minimize} \qquad & {\textstyle\sum_{o\in\objectspace}} z_o
	\\
	\text{subject to}\qquad & q \in \feasible
	\\
	& z_o \ge  \textstyle\sum_{t\in\secretspace\setminus\{s\}}  \pi_t C_{t,o}    &\forall o\in \objectspace
	\\
	& z_o \ge \pi_s q_o  &\forall o\in \objectspace
\end{align*}


%% file: 5_unknown-prior.tex
\section{Optimizing $q$ for an unknown prior}\label{sec:unknown-prior}

In the previous section we discussed finding the optimal $q$ then the prior
$\pi$ (i.e. the user profile) is fixed. In practice, however, we often do not
know $\pi$ or we do not want to restrict to a specific one. The natural goal
then is to design our system wrt the \emph{worst} possible prior.
The QIF theory will again offer guidance in selecting $q$; this time
we choose the one that minimizes $C^q$'s \emph{capacity}, i.e. we minimize its maximum
leakage wrt all priors.

\subsection{Exact-guessing adversary}

Starting with the problem of optimizing $q$ wrt an exact-guessing adversary,
in this section we make two observations. First, we show that finding an
optimal $q$ is possible either via simple convex combinations of rows (provided 
that such solutions are feasible), or via linear programming.
Second, and more important, we show that
selecting $q$ solely wrt this adversary is a poor design choice, since many
values are simultaneously optimal, although their behavior is not equivalent
for other adversaries.

Recall that the natural choice of $q$ for a fixed prior
(\autoref{sec:fixed-prior}) was $\OptQ\Wide =\WtoS_{\neg s}C$, which is a
convex combination of the rows $C_{\neg s}$. It turns out that for an unknown
prior and an exact-guessing adversary we can choose much more freely:
\emph{any} convex combination of the rows $C_{\neg s}$ minimizes
capacity. To understand
this fact, consider $\GEx$-capacity and recall that $\mulcapacity_\GEx(C^q)$ is
given by the sum of the column maxima of the channel (\autoref{thm:mbc}).
Adding a new row cannot decrease the column maxima, hence
$\mulcapacity_\GEx(\Ch^q) \ge \mulcapacity_\GEx(C^q_{\neg s})$. Moreover,
achieving equality is trivial: setting $q$ to any convex combination of rows
means that no element of $q$ can be strictly greater than all corresponding
elements of $C^q_{\neg s}$, hence $\Ch^q$ and $C^q_{\neg s}$ will have the
exact same column maxima. This brings us to the following result:

\begin{restatable}{proposition}{CHOptimalForExact}\label{prop:unknown-exact-optimal}
	For all $\Ch$,
	any $\OptQ \in \ConvHull{C_{\neg s}}$
	minimizes capacity for exact-guessing adversaries, that is
	\[
		\mulcapacity_\LorG(\channel^{\OptQ})
		\Wide\le
		\mulcapacity_\LorG(\channel^{q})
		\;\quad\forall q\in\distset\objectspace, \LorG\in\{ \GEx,\LEx \}
		~.
	\]
	Moreover, it holds that $\mulcapacity_\LorG(\Ch^{\OptQ}) = \mulcapacity_\LorG(C_{\neg s})$.
\end{restatable}

A direct consequence of \autoref{prop:unknown-exact-optimal} is that
any convex combination of the rows $C_{\neg s}$ that happens to be feasible,
that is any $\OptQ \in \feasible \cap \ConvHull{C_{\neg s}}$,
is an optimal solution for an exact-guessing adversary.
Note that this is a sufficient
but not necessary condition for optimality. In \autoref{sec:unknown-s-dist}
we see that solutions outside the convex hull can be also optimal.

On the other hand, there is no guarantee that any such solution exists,
it could very well be the case that no convex combination of rows is feasible.
In this case, we can still compute an optimal solution, as follows:

\begin{proposition} \label{prop:unknown-exact-lp}
	The optimization problem
	\begin{equation}
		\OptQ \Wide{:=} \argmin_{q\in\feasible} \mulcapacity_\LorG(\Ch^q)	
		~,\quad \LorG\in\{ \GEx, \LEx \} ~,
		\label{eq9324}
	\end{equation}
	can be solved in polynomial time via linear programming.
\end{proposition}

The linear program for $\GEx$ is given below:
\begin{align*}
	\text{minimize} \qquad & {\textstyle\sum_{o\in\objectspace}} z_o
	\\
	\text{subject to}\qquad & q \in \feasible
	\\
	& z_o \ge  \textstyle\max_{t\in\secretspace\setminus\{s\}}  C_{t,o}    &\forall o\in \objectspace
	\\
	& z_o \ge q_o  &\forall o\in \objectspace
\end{align*}

\autoref{prop:unknown-exact-optimal} states that a large set of choices for $q$ are all equivalent
from the point of view of the exact-guessing capacity.
However, this is not true wrt other types of adversaries, as the following
example demonstrates:
\[
C^q \Wide=
\begin{array}{ccc}
	& o_1 & o_2\\
	\cline{2-3}
	\multicolumn{1}{c|}{s} & q_1 & \multicolumn{1}{c|}{q_2} \\
	\multicolumn{1}{c|}{s_1} & 1  & \multicolumn{1}{c|}{0} \\
	\multicolumn{1}{c|}{s_2} & 0 & \multicolumn{1}{c|}{1} \\
	\cline{2-3}
\end{array}
\]
Here $C_{\neg s}$ is a deterministic channel with only 2 secrets that are
completely distinguishable. For instance, two websites, without any
obfuscation mechanism, having distinct page sizes.
From \autoref{thm:mbc} we can compute
$\mulcapacity_\GEx(C^q_{\neg s}) = 2$.
Since the rows $s_1,s_2$ are already maximally distant, the choice of $q$ is
irrelevant. For \emph{any} $q$, the rows $s_1,s_2$ will still be maximally
distant, giving $\mulcapacity_\GEx(C^q) = 2$.

Although $q$ does not affect $\mulcapacity_\GEx$, this does not mean that the choice of
$q$ is irrelevant for the website $s$. Setting $q = (1, 0)$ we make $s$
indistinguishable from $s_1$ but completely distinguishable from $s_2$.
Conversely, $q = (0, 1)$ makes $s$ indistinguishable from $s_2$ but
completely distinguishable from $s_1$. Finally, $q = (\nicefrac{1}{2},
\nicefrac{1}{2})$ makes $s$ somewhat indistinguishable from both $s_1$ and
$s_2$; intuitively the latter seems to be a preferable choice, but \emph{why}?

To better understand how $q$ affects the security of this channel we should
study the difference between an exact-guessing and an $s$-distinguishing
adversary. For the former, recall that $\mulcapacity_\GEx$ is always given by
a uniform prior $\Uni$. For such a prior, an adversary who guesses the secret
after observing the output will always guess $s_1$ after seeing $o_1$
(because $s_1$ always produces $o_1$) and $s_2$ after seeing $o_2$,
independently from the values of $q$. So intuitively $q$ does not affect this
adversary at all, which is the reason why
$\mulcapacity_\GEx(C^q) = 2$ for any $q$.

However, for an $s$-distinguishing adversary, the situation is very
different. When $q = (1,0)$ we can use \autoref{thm:p-guessing-capacity} to
compute $\mulcapacity_{\GS}(\Ch^q) = 2$, given by the prior $\Uni^{s,s_2}$;
for this prior the adversary can trivially infer whether the secret is $s$ or
$\neg s$ after the observation. But for $q' =
(\nicefrac{1}{2},\nicefrac{1}{2})$ we get $\mulcapacity_{\GS}(\Ch^{q'}) =
\nicefrac{3}{2}$, realized by both $\Uni^{s,s_1}$ and $\Uni^{s,s2}$. The
system provides non-trivial privacy even in the worst case; $s$ and $\neg s$
can never be fully distinguished.

The discussion above suggests that maximizing the exact-guessing capacity by
itself does not fully guide us in choosing $q$; it is meaningful to also optimize
wrt an $s$-distinguishing adversary. 
In fact, in the next section we see that optimizing
wrt \emph{both} $s$-distinguishing and exact-guessing adversaries
simultaneously is sometimes possible.



\subsection{$s$-distinguishing adversary}\label{sec:unknown-s-dist}

We turn our attention to optimizing $q$ wrt an $s$-distinguishing adversary for an unknown prior.
We already know from \autoref{thm:p-guessing-capacity} (for $P = \{s\}$) that
both $\GS$ and $\LS$-capacities depend on the
maximum distance $\MaxDist{C_s}{C_{\neg s}}$ between the row $C_s$ and all
other rows of the channel. In other words, the capacity is related to the radius of the
smallest $\lnorm{1}$-ball centered
at $\Ch_s$ that contains $\Ch_\secretspace$.

This gives us a direct way of optimizing $q$ wrt $\GS,\LS$-capacity by a
solving a \emph{geometric} problem known as the \emph{smallest enclosing
ball} (SEB): find a vector that minimizes its maximal distance to a set of known
vectors, or equivalently find the smallest ball that includes this set.

Interestingly, it turns out that in one particular case the SEB solution is
guaranteed to be simultaneously  optimal wrt an exact-guessing adversary.  This
happens in the \emph{unconstrained} case $\feasible = \distset\objectspace$,
that is when any solution $q$ is feasible, as stated in the following
result.

\begin{restatable}{theorem}{Minball}\label{thm-minball}
	For all $C:\secretspace\to\objectspace$, any distribution given by
	\[
		\OptQ
		\Wide\in
		\argmin_{q\in\feasible} \;\MaxDist{q}{C_{\neg s}}
	\]
	gives optimal capacity for $s$-distin\-guishing adversaries:
	\begin{align*}
		\mulcapacity_\LorG(\channel^{\OptQ})
			&\Wide\le
			\mulcapacity_\LorG(\channel^{q})
			~,\quad \LorG\in\{ \GS, \LS \},
			q\in\feasible
			~.
	\end{align*}
	Moreover, if
	$\feasible = \distset\objectspace$, then $\OptQ$
	is simultaneously optimal for exact-guessing adversaries,
	i.e. for $\LorG\in\{ \GEx, \LEx \}$.
\end{restatable}
Note that the solution $\OptQ$ obtained from the above
result might lie \emph{outside}
the convex hull of $C_{\neg s}$.
So, in the unconstrained case, the simple optimality
conditions of
\autoref{prop:unknown-exact-optimal} are \emph{not} met,
yet the resulting solution is still guaranteed to be optimal
also for exact-guessing adversaries.

The smallest enclosing ball problem is discussed in the next section,
showing that it can be solved in linear time on $|\secretspace|$ for fixed $|\objectspace|$,
or in polynomial time on $|\secretspace|\cdot |\objectspace|$.

%% file: 6_seb.tex

 \section{The SEB problem }
 \label{sec:seb}

The following is known as the smallest enclosing ball (SEB) problem \cite{seb_first_paper}:
given a finite subset $S \subseteq M$ of
some metric space $(M,d)$, find the smallest
ball $B_r(x), x\in M, r\in\Reals$ that contains $S$.
In this paper,
the goal is to solve the problem for $M= \feasible$ and $d= \lnorm{1}$
(see \autoref{sec:unknown-prior}).

\subsubsection*{Euclidean norm}
The problem is well-studied for $(\Reals^m, \lnorm{2})$ \cite{on_to_SEB}. It has been shown that
the solution is always unique and belongs to the convex hull of $S$. For any
fixed $m$, it can be found in linear time on $n = |S|$. However, the
dependence on $m$ is exponential\footnote{A sub-exponential algorithm
does exist, but still its complexity is larger than any polynomial.}.
Nonetheless, as we discuss below,  approximation algorithms also exist that
can compute a ball of radius at most $(1+\epsilon)r^*$, where $r^*$ is the optimal
radius and $\epsilon > 0$,
in time linear on both $m$ and $n$.

\subsubsection*{Manhattan norm}
For $(\Reals^m, \lnorm{1})$ the problem is much less studied.
The solution is \emph{no longer unique}, due to the fact that $\lnorm{1}$-balls
have straight-line segments in their boundary.
Moreover, somewhat surprisingly, \emph{none} of the solutions is guaranteed to be
in the convex hull of $S$.

Similarly to the Euclidean case, for fixed $m$ the problem can be solved in time linear on $n$,
using the isometric embedding of
$(\Reals^m, \lnorm{1})$
into $(\Reals^{2^m}, \lnorm{\infty})$. In contrast to the Euclidean case,
however, the problem can be solved in polynomial time on both $n$ and $m$ via
linear programming.

\subsubsection*{Probability distributions}
Our case of interest
is $(\feasible, \lnorm{1})$
for $\feasible\subseteq\distset\objectspace$,
the set of constrained probability distributions
over some set finite $\objectspace$, under the $\lnorm{1}$-distance.
Note that $\distset\objectspace$ is a subset of $\Reals^m$ for $m = |\objectspace|$.
However, solving the
$(\Reals^m, \lnorm{1})$-SEB problem does not immediately yield a solution for
$(\feasible, \lnorm{1})$-SEB, since the center of the optimal ball
might lie outside $\feasible$, or even outside $\distset\objectspace$.
This problem is studied in the following sections.

\subsection{Linear time solution for fixed $m$}

We start from the fact that the
$(\Reals^{d}, \lnorm{\infty})$-SEB problem admits a direct solution:
given a set $S \subset \Reals^{d}$, denote by $S^\top,S^\bot \in\Reals^d$ the vectors
of component-wise maxima and minima:
\begin{equation}\label{eq9853}
	S^\top_i = \max_{x \in S} x_i
	\qquad
	S^\bot_i = \min_{x \in S} x_i
	\qquad
	i \in \{1,\ldots,d\}
	~.
\end{equation}
It is easy to see
that the optimal radius is $r^* = \frac{1}{2}\|S^\top - S^\bot\|_\infty$, and the (non-unique) optimal
center is $x^* = \frac{1}{2}(S^\top + S^\bot)$.

Moving to the $(\Reals^m, \lnorm{1})$-SEB problem, we use a well-known embedding $\varphi: \Reals^m\to\Reals^{2^m}$
for which it holds that $\|\varphi(x) - \varphi(x')\|_\infty = \|x - x'\|_1$.
Using the fact that $\varphi$ is invertible, an optimal solution $(x^*, r^*)$
for $(\Reals^{2^m}, \lnorm{\infty})$-SEB
can be directly translated to an optimal solution
$(\varphi^{-1}(x^*), r^*)$ for $(\Reals^m, \lnorm{1})$-SEB.

Turning our attention to our problem of interest,
the  $(\feasible, \lnorm{1})$-SEB case is a bit more involved.
We can still use the same embedding $\varphi$, but
an optimal solution $(x^*, r^*)$ for $(\Reals^{2^m}, \lnorm{\infty})$-SEB cannot
be translated to our problem since
$\varphi^{-1}(x^*)$ is not guaranteed to be a probability distribution;
in fact no solution of radius $r^*$ is guaranteed to exist at all.
Writing $\varphi(S)$ for $\{ \varphi(x) \ |\ x\in S \}$,
essentially what we need is to solve the
$(\varphi(\feasible), \lnorm{\infty})$-SEB problem;
in other words to impose that the solution is the translation of a
feasible probability distribution.

The first step is to compute the vectors $\varphi(S)^\top, \varphi(S)^\bot$
of component-wise maxima and
minima for each translated vector.
Although we cannot directly construct the solution from these vectors
(as we did for $(\Reals^{d}, \lnorm{\infty})$-SEB), the
key observation is that these two vectors alone represent the maximal distance to the whole $S$,
because $\forall y\in \Reals^m$:
\begin{align*}
	&\max_{x \in S} \| y - x \|_1
	\Wide=
	\max_{x \in S} \| \varphi(y) - \varphi(x) \|_\infty
	\\
	&\Wide=
	\max \{ \| \varphi(y) - \varphi(S)^\top \|_\infty,  \| \varphi(y) - \varphi(S)^\bot \|_\infty\}
	~.
\end{align*}


Then, we exploit the fact that $\varphi$ is a \emph{linear map};
more precisely
\[
	\varphi(x) = x \Phi
	~,
\]
where $\Phi$ is a $m\times 2^m$ matrix, having one column for each bitstring $b$ of size $m$,
defined as $\Phi_{i,b} = (-1)^{b_i}$.
This allows us to solve the $(\feasible, \lnorm{1})$-SEB problem
via \emph{linear programming}: we use $x\in\Reals^m$ as variables,
imposing the linear constraints $x \in \feasible$.
Moreover, we ask to minimize the $\lnorm{\infty}$-distance between
$x\Phi$ and $\varphi(S)^\top,\varphi(S)^\bot$, two vectors that we have computed in advance.
The program can be written as:
\begin{align*}
	\text{minimize}\qquad z
	\\
	\text{subject to}\qquad x &\in \feasible
	\\
	z &\ge \phantom{-}\varphi(S)^\top_b - (x\Phi)_b	& \forall b \in \{0,1\}^m
	\\
	z &\ge - \varphi(S)^\bot_b + (x\Phi)_b & \forall b \in \{0,1\}^m
\end{align*}

Note that $\varphi(S)^\top,\varphi(S)^\bot$ can clearly be computed in $O(n)$ time.
Given these vectors, the whole linear problem does not depend on $n$ (because it does
not involve $S$). For fixed $m$, solving the linear program takes constant time, which
implies the following result.

\begin{theorem}
	The $(\distset\objectspace, \lnorm{1})$-SEB problem can be solved in
	$O(n)$ time for any fixed $m$.
\end{theorem}

\subsection{Polynomial time solution for any dimension} \label{seb_exact}

In contrast to the Euclidean case, the dependence on $m$ for
$(\feasible, \lnorm{1})$-SEB is polynomial.
This is because the objective function $\max_{y \in S} \| x-y \|_1$,
        can be turned into a linear one using auxiliary variables.
We use variables $w_{y,i}$
to represent $|x_i - y_i|$, and a variable $z$
to represent $\max_{y \in S} \sum_i |x_i - y_i| = \max_{y \in S} \sum_i w_{y,i} $.

The linear program can be written as:
\begin{align*}
	\text{minimize}\qquad z
	\\
	\text{subject to}\qquad x &\Wide\in \feasible
	\\
	w_{y,i} &\Wide\ge x_i - y_i
		&& \forall y\in S, i\in\{ 1,\dots,m\} \\
	w_{y,i} &\Wide\ge y_i - x_i
		&& \forall y\in S, i\in\{ 1,\dots,m\} \\
	z &\Wide\ge \textstyle\sum_i w_{y,i}
		&& \forall y\in S
	\\
\end{align*}

\begin{theorem}
	The $(\feasible, \lnorm{1})$-SEB problem can be solved in
	polynomial time, using a linear program with $O(nm)$ variables and $O(nm)$ constraints.
	\label{thm:seb-exact}
\end{theorem}

\subsection{Approximate solutions} \label{seb_approx}

The $(\feasible, \lnorm{1})$-SEB problem can be approximated
by solving the
$(\Reals^m, \lnorm{2})$-SEB problem for which several algorithms exist,
and then projecting the solution to $\feasible$.

For an exact (Euclidean) solution, there are known algorithms
claimed to handle dimensions up
to several thousands \cite{Fischer03fastsmallest-enclosing-ball}. Note that
an exact solution is unique and is guaranteed to lie within the convex hull of $S$.
Hence, when applied to distributions, the solution is guaranteed to be a
distribution.
Moreover, \autoref{prop-conv-hull-optimal} guarantees that if the solution
is feasible, it will be
optimal for an exact-guessing adversary, although it will not be optimal for
an $s$-distinguishing one.

Furthermore, there are several approximate algorithms that run in linear time or
even better (see \cite{DBLP:journals/corr/abs-1904-03796} for a recent work
which provides several references). Their solution is not guaranteed
to be a probability distribution, so a projection to $\feasible$ will be needed.

In the experiments
of \autoref{sec:exp}, we call 
\emph{SEB exact} the solution of \autoref{thm:seb-exact},
and \emph{SEB approx} the solution obtained via a linear-time approximation
algorithm.

%% file: 8_evalution.tex
\section{Use-case: Website Fingerprinting } \label{sec:exp}

In this section we apply the optimization methods described in previous sections to defend against Website Fingerprinting (WF) attacks and evaluate their performance. In a WF attack the adversary observes the encrypted traffic pattern between a user and a website and tries to infer which website the user is visiting. This is a particularly interesting attack when the adversary cannot directly observe the sender of the intercepted packets, for example when traffic is sent through the Tor network \cite{tor} for which a series of WF attacks have been proposed \cite{wf_tor_1 ,wf_tor_2}.

\subsection{Setup} \label{exp:setup}

Consider a news website covering ``controversial'' topics, prompting an
adversary to target it and attempt to identify its readers via WF.
To defend against such an attack,
the administrator of this website would like to make its responses as indistinguishable as possible from other news sites, so that WF becomes
harder.

For simplicity, in our evaluation we consider that only one request is
intercepted by the adversary and only the size of the encrypted response is
observed. The administrator's goal is to try to imitate other
websites  by producing pages that are similar in size.  Such target websites
produce responses according to distributions which are assumed to be known both
to the administrator and to the adversary. 

For our evaluation, we started by identifying
the top 5 (in traffic) news website from 40 countries\footnote{According to \url{www.similarweb.com}.}, leading to a total
of 200 sites, of which one is selected as the defended site $s$.
Then, we crawled these sites and measured the size of the received pages, rounded to the closest KB, creating a distribution over the page sizes. The biggest page size is 300KB, hence the output space of the distribution is 1KB, 2KB,.., 300KB.

Note that obtaining an accurate distribution of an \emph{average} request
requires knowledge of how the traffic is distributed across the different
pages of the site. For instance, 
the probability of visiting a page typically decreases as the user navigates deeper into the site: it is more likely for users to read 
the headlines on the homepage than to access a page several links deeper.

For our evaluation, 
we simulate visitors by randomly following 10 links of every page up to a
maximum depth of 4. A probability distribution over the page sizes is constructed
by assuming that a visitor access each depth with the following probabilities:
\begin{itemize}
    \item home page: 0.3
    \item 1-click depth: 0.25
    \item 2-clicks depth: 0.2
    \item 3-clicks depth: 0.15
    \item 4-clicks depth: 0.1
\end{itemize}
while the probability of accessing pages within the same layer is uniform.


Note that defending against WF when we can control  only a single site is quite challenging: the selected 200 sites are quite different from each other, so
we cannot imitate all of them simultaneously. Instead, we assume that the administrator of $s$ selects a moderate subset of 19 sites close in distance to $s$,
leading to a system of 20 secrets that we try to minimize its leakage.\footnote{The selected sites can be found in the \Cref{appendix:selected_sites}.}


To successfully hide $s$ among the other 19 sites, the administrator needs to find $q$ (i.e. a distribution over the page sizes) in order to reduce the leakage (or capacity) of $C$. Then, he needs to modify his site so that the served pages follow this distribution.

Note, however, that not all distributions $q$ are feasible for the defended site since the administrator still needs to respect the site's existing content.
While increasing a page size is usually straightforward via padding\cite{cherubin2017website}, decreasing it may not be feasible without affecting the content.
In the following experiments we take this issue into account by enforcing a \emph{non-negative padding} constraint $\feasible$, which is discussed below.


\subsubsection*{Priors}
Let us first describe the priors that we are about to use in the following experiments:

\begin{itemize}
    \item \emph{Uniform}: \textbf{u}, that is probability $1/20$ for each site.
    \item \emph{Traffic}: Based on the monthly visits of each website.
\end{itemize}

In practice, the adversary may have suspicions regarding the user's location. For example, a user living in the EU would probably not have a regular interest in reading the daily news from Brazil. 
Instead, she may frequently visit news websites from their own country or neighboring countries. 

To simulate this, considering that $s$ is a Romanian site, we create the following priors:

\begin{itemize}
    \item \emph{Eastern}: Proportional to the country's population if the country is in the Eastern Bloc of EU, 0 otherwise
    \item \emph{Ro-Slo}: Uniform if the site is hosted in either Romania or Slovakia, 0 otherwise.
    \item \emph{Ro-Hu}:  Uniform if the site is hosted in either Romania or Hungary, 0 otherwise.
\end{itemize}

\subsubsection*{Baseline Methods}\label{baseline_methods}
We are about to compare our proposed solutions with the following natural approaches:
\begin{itemize}
    \item \emph{No Defense}: Setting $q=C_s$.
    \item \emph{Average}: For each page size, calculate the average probability across all the other sites. Used only in the experiments for unknown prior. 
    \item \emph{Weighted \emph{Average} on prior $\pi$}: Similar to \emph{Average}, but now assign weights to each row based on the prior. 
    \item \emph{Copy}: Emulates another site, i.e. setting $q=C_{t'}$ for some site $t'$. If $\pi$ is known choose the $t'$ with the biggest $\pi_{t'}$. Otherwise, choose the $t'$ that offers the minimum capacity.
    \item \emph{Pad}: Pad each page size deterministically to the next multiple of 5KB. 
\end{itemize}

\subsubsection*{Feasible Solutions} \label{exp:feasible_solution}

As previously discussed, when searching for an optimal distribution $q$ of page sizes, we need to take into account that not all distributions are feasible in practice,
which is expressed by enforcing feasibility constraints $\feasible$.
In our use case, the constraints arise from the fact that page size can be easily increased via padding, but not decreased.
For example, say that $q$ dictates that we should produce a page size of 5KB with probability 0.2. If all actual pages of our site are 10KB or larger, then producing a page of 5KB with non-zero probability is impossible.

 To apply the optimization methods of previous sections we need to express $\feasible$ in terms of linear inequalities, which is done as follows.
Let $\hat{q}$ be the site's \emph{original} size distribution,  without applying any defense.
To express our non-negative padding constraint, we create a matrix $T_{o,o'}$, denoting the probability
that we move from $o$ to $o'$ for all pairs $o,o' \in\objectspace $. Then, a solution $q$ is feasible iff it can obtained from $\hat{q}$ via
a transformation table $T$ that only moves probabilities from smaller to larger observables,
which is expressed by the following linear constraints:
\begin{align*}
    T_{o,o'} &\ge 0  & \forall{o, o'}   \\
    T_{o,o'} &= 0   & \forall{o > o'}     \\
	\hat{q}_o &= \textstyle\sum_{o'}{T_{o,o'}} &\forall o \\
    q_o &= \textstyle\sum_{o'}{T_{o',o}}  &\forall o 
\end{align*}
This approach not only expresses the feasibility of $q$, but also provides the matrix $T$ which can be directly used as a padding strategy.
When we receive a request for a page with size $o$, we can use the row $T_{o,\cdot}$ (after normalizing it),
as a distribution to produce a padded page. The constraints guarantee that all sizes produces by that distribution will not be smaller than $o$.


\subsection{Experiments for a fixed prior $\pi$}
This section evaluates the solutions discussed in Section \ref{sec:fixed-prior}.

Starting from the \emph{s-distinguishing adversary}, recall that 
the weighted (by the prior) average of the other rows $q = \WtoS_{\neg s}C$ is guaranteed to have no leakage (\autoref{prop-conv-hull-optimal}).
In our experiment, however, this simple solution cannot be directly applied  since it violates the feasibility constraints.
Still, we can use the LP solution of \autoref{optimazation_prop2} to obtain the optimal feasible solution.
The results are shown in Table \ref{table:fixed-prior-s-dist}; we can see that, although $q = \WtoS_{\neg s}C$
itself is not feasible, we can actually find a feasible solution with no leakage at all in all cases,
except from Ro-Hu in which the leakage is slightly larger than 1. 

\begin{table}[b] 
  \centering
    \caption{Leakages for each prior using Proposition \ref{optimazation_prop2} (known $\pi$, s-distinguishing adversary)}
  \begin{tabular}{@{}cc@{}}
    \toprule
    Prior  & Leakage\\
    \midrule
    Uniform & 1 \\ 
 \hline
 Traffic & 1 \\ 
 \hline
  Eastern & 1 \\ 
   \hline
 Ro-Slo & 1 \\ 
  \hline
  Ro-Hu & 1.03 \\ 
   
    \bottomrule
  \end{tabular}

   \label{table:fixed-prior-s-dist}
\end{table}

\begin{table*}[h]
  \centering
    \caption{Leakage and posterior vulnerability (in parenthesis) for each method and prior (known $\pi$, exact-guessing adversary)}
  \begin{tabular}{@{}ccccccc@{}}
    \toprule
    Method  & Uniform prior   & Traffic prior & Eastern prior & Ro-Slo prior & Ro-Hu prior \\
    \midrule
      \rowcolor{gray!30}
     Optimal (Prop. \ref{optimazation_prop2}) &  8.78 (0.44) & 2.29 (0.61) & 1.01 (0.44) & 1.75 (0.58) & 1.04 (0.52) \\

    No Defense  & 9.16 (0.46) & 2.29 (0.61) & 1.69 (0.74) & 2.40 (0.80) & 1.76 (0.88)\\
   
    Weighted Average & 9.02 (0.45) & 2.29 (0.62) & 1.21 (0.53) & 1.90 (0.63) & 1.04 (0.52) \\
    
    Copy  & 8.8 (0.44) & 2.29 (0.61) & 1.43 (0.63) & 1.78 (0.59) & 1.04 (0.52)\\
    
    Pad & 9.26 (0.46) & 2.29 (0.61) & 1.86 (0.81) & 2.55 (0.85) & 1.89 (0.94) \\
   
    \bottomrule
  \end{tabular}

  \label{table:fixed-prior-exact-leak-vuln}
\end{table*}

Moving to the \emph{exact-guessing adversary}, we again use the LP of \autoref{optimazation_prop2} to find
the optimal solution (since the     weighted average $q = \WtoS_{\neg s}C$ is neither feasible nor optimal).
\Cref{table:fixed-prior-exact-leak-vuln} shows the leakage for each prior and method, as well as
the \emph{posterior vulnerability} (in parenthesis) which is helpful to interpret the results.
For instance, the optimal solution for the Ro-Slo prior gives a leakage of 1.75 and posterior vulnerability
of 0.58, meaning that the adversary can guess the secret with probability $\nicefrac{0.58}{1.75} = 0.33$
a priori, and his success probability increases to $0.58$ after observing the output of the system.
Note that this adversary is much harder to address (by controlling only a single row of the channel) than the s-distinguishing,
hence it is impossible to completely eliminate leakage in most cases, but we can still hope for a substantial improvement
compared to having no defense at all.

The results show that our approach offers the least leakage and the smallest posterior vulnerability across all priors. Among the other options, the natural choice of (projected) \emph{Weighted Average} offers comparable results for some priors (Traffic, Ro-Hu) , but notably worse for others (for example $\approx 20\%$ more leakage on the Eastern prior and $\approx 10\%$ more leakage on the Ro-Slo prior). 

\emph{Copy} seems to be an interesting alternative for some priors but offers a solution with increased leakage on the  case of the Eastern prior (1.43 over  1.01 of \Cref{optimazation_prop2}).
This prior is the only one in which $s$ has a significantly larger probability (0.43) than any other site $t$; in this case, if we can make the evidence of any
observation $o$ smaller than our a priori belief, that is $\nicefrac{C_{t,o}}{C_{s,o}} \le \nicefrac{\pi_s}{\pi_{t}}$, then the rational choice would be
to guess $s$ for any observation $o$, and the system would have no leakage at all. This is indeed achieved by the optimal solution.
If we choose to \emph{Copy}, however,  the site $t$ with the largest (other than $s$) prior, we might not necessarily achieve this goal. This solution does make $s$ and
$t$ indistinguishable, but we might now produce certain observations with very small probabilities, making $s$ and $t$ distinguishable from some of the the remaining
sites, which explains why \emph{Copy} performs worse than \emph{Optimal} for this prior.


Another somewhat surprising observation is that padding turns out to be more \emph{harmful} than no defense at all.
The reason is that padding relies on every site using it simultaneously, so that different size observations become
identical when mapped to the same padded value. However, if only $s$ pads, then its observations can become
even more distinguishable than before, since only that site will always report sizes that are multiples of 5KB.


\subsection{Experiments for an unknown prior $\pi$}
This section aims to assess the solutions presented in Section \ref{sec:unknown-prior}.

\subsubsection*{exact-guessing adversary}
We discussed that any convex combination of the remaining rows of $C$ will yield the same capacity, provided that the solution is feasible, since each column maximum is being retained. But a convex combination that complies with the constraint might not exist at all. To overcome the hurdle,  we can use LP (Proposition \ref{prop:unknown-exact-lp}).

\Cref{table:unkown-prior-exact} shows that the \emph{Optimal (Prop. \ref{prop:unknown-exact-lp})} offers the best capacity of 8.78. In fact, any other method that is a convex combination of the remaining rows (i.e. \emph{Average}, \emph{Copy}) would have had the same result, but the constraints led to a slightly increased capacity. Observe that for the exact guessing adversary this capacity is the same as the leakage on the uniform prior (\Cref{table:fixed-prior-exact-leak-vuln}), which is the worst prior (for the defender) for the exact guessing adversary, as discussed in  \Cref{prelim}.

Nonetheless, \Cref{sec:unknown-s-dist} discussed that we could also use the solution of SEB (for the unconstrained case). Even though constraints do exist in this experiment, \Cref{table:unkown-prior-exact} shows that we get the same capacity as the \emph{Optimal (Prop. \ref{prop:unknown-exact-lp})}. Interestingly, neither SEB nor Proposition \ref{prop:unknown-exact-lp} provide a solution that belongs to the convex hull of \(C\).

\begin{table}[h] 
  \centering
    \caption{Capacity for each method (unknown $\pi$, exact-guessing adversary)}
  \begin{tabular}{@{}ccc@{}}
    \toprule
    Method  & Capacity \\
    \midrule
      \rowcolor{gray!30}
    Optimal  (Prop. \ref{prop:unknown-exact-lp})  & 8.78 \\ 
 \hline
 \rowcolor{gray!30}
    SEB exact  & 8.78 \\
    \hline 
 No Defense  & 9.14 \\ 
\hline
Average &  9.02 \\ 
 \hline
 Copy & 8.8  \\ 
 \hline
 Pad & 9.25  \\ 
    \bottomrule
  \end{tabular}

   \label{table:unkown-prior-exact}
\end{table}

\subsubsection*{s-distinguishing adversary}

 Recall that we discussed why the problem reduces to SEB, offering two solutions, an exact but slow one (\emph{SEB exact}) and a faster approximation (\emph{SEB approx.}) resp. in  \Cref{seb_exact} and  \Cref{seb_approx}.

\Cref{table:unkown-prior-s_dist} shows that \emph{SEB exact} offers the best capacity: it decreases the capacity from $\approx 1.8$ (\emph{No Defense}) to $1.56$. \emph{Average} does show some improvement over \emph{No Defense}, while \emph{Pad} yields results that are too easily distinguishable. On the other hand, \emph{SEB approx.} is technically the second-best choice, although the result is nearly identical to \emph{Average}. 

\begin{table}[h] 
  \centering
    \caption{Capacity for each method (unknown $\pi$, s-distinguishing adversary)}
  \begin{tabular}{@{}cc@{}}
    \toprule
    Method  & Capacity \\ 
    \midrule
    \rowcolor{gray!30}
    SEB exact  & 1.56 \\ 
    \hline
    \rowcolor{gray!12}
    SEB approx.  & 1.653 \\  
    \hline
    No Defense   & 1.79 \\ 
    \hline
    Average & 1.659 \\ 
    \hline
    Copy & 1.79 \\ 
    \hline
    Pad & 1.9 \\ 
    \bottomrule
  \end{tabular}
  \label{table:unkown-prior-s_dist}
\end{table}

\subsection{Attack Simulation} \label{exp:attack_sim}
In this section we simulate an actual attack, to estimate the attacker's accuracy, using the s-distinguishing adversary, as it is directly applicable in the WF scenario. In a real-world attack the defender will probably be oblivious for the attacker's $\pi$ and hence we use the solutions for the unknown prior. 

We train a Random Forest classifier to estimate the target site among the others based on the observed page size, sampling pages from each site $t$ according to $C_t$.
While this scenario involves an unknown prior, it is essential to establish one solely for the attacker, in order to decide how many pages to sample from each site. In other words, for each site $t$ the number of sampled pages depends on $\pi_t$. For $s$, we set $\pi_s = 0.5$ for all the following experiments, to capture the boolean nature of this adversary. Note that $\pi$ is only known to the adversary; it is not disclosed to the defender, who for that reason uses prior-agnostic methods.

In order to train the classifier for $s$, we begin by requesting page sizes according to $C_s$ and record the actually reported page sizes from the server (derived from $T$), which are then used in the training process. Essentially, this simulates the behavior of the server of $s$; a user requests a page (selecting not uniformly randomly, but according to $C_s$), then the server calculates the page size $i$, pads it according to $T_i$, and finally sends the padded page back to the user.

Finally, as per standard procedure, $80\%$ of the data is used for training and the remaining $20\%$ for testing.

\begin{figure*}[h]
\centering
\begin{subfigure}{.42\textwidth}
  \centering
  \includegraphics[width=\linewidth]{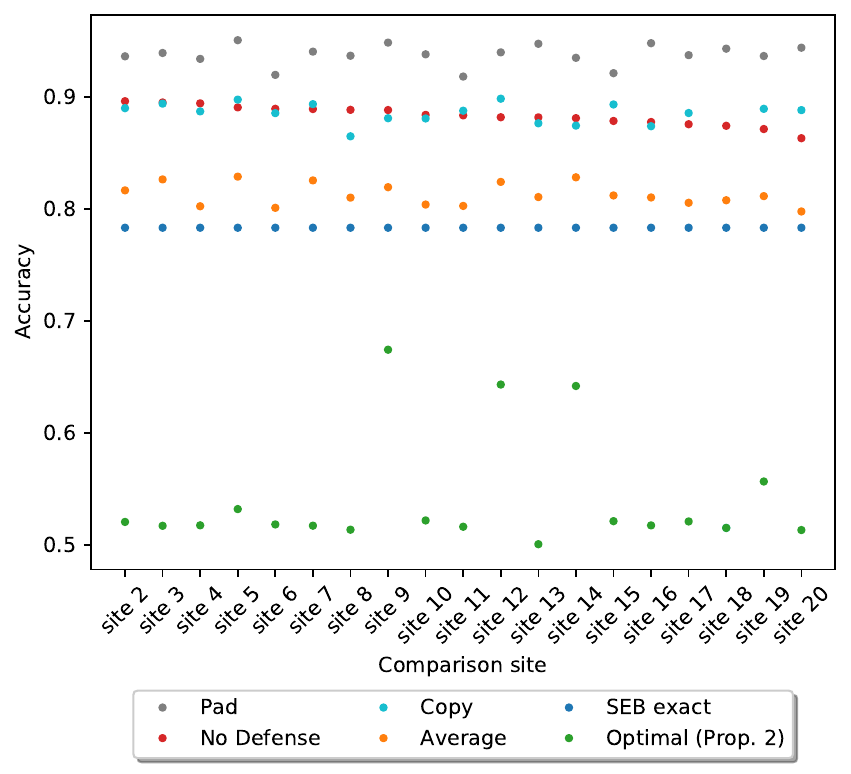}
  \caption{1-on-1: prior is $0.5$ on $s$ and $0.5$ on another site \newline }

  \label{fig:accuracy_by_duel}
\end{subfigure}%
\begin{subfigure}{.42\textwidth}
  \centering
  \includegraphics[width=\linewidth]{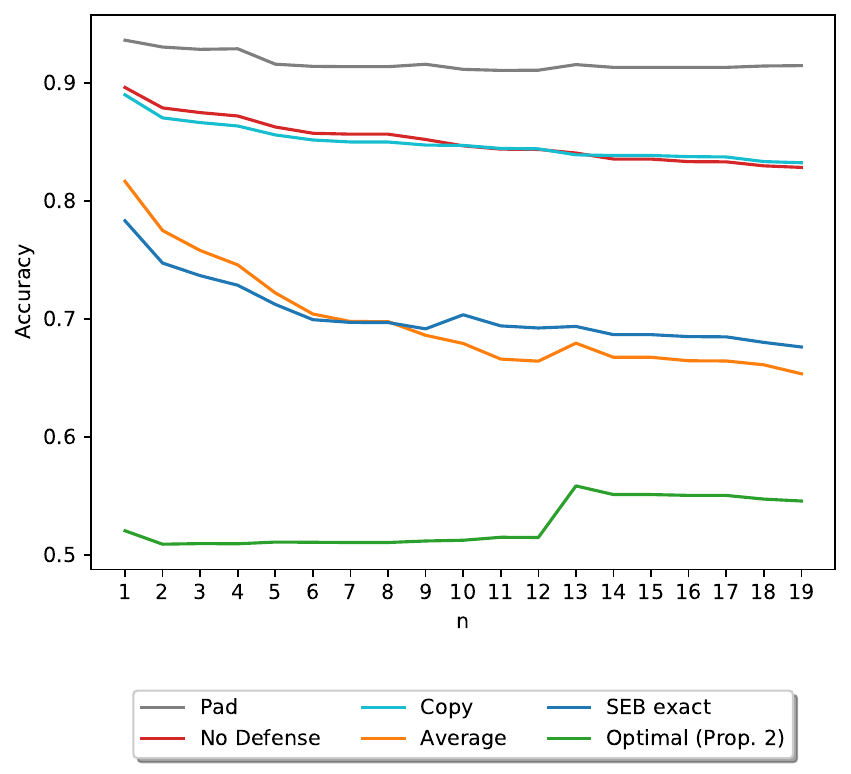}
  \caption{1-on-n: prior is $0.5$ on $s$ and the remaining $0.5$ is divided across $n$ sites according to their traffic }
  
  \label{fig:accuracy_by_visits}
\end{subfigure}
\caption{Attacker's accuracy} 
\label{fig:50_50_prior}
\end{figure*}

Note that having a prior where $\pi_s > 0.5$ would essentially decrease the attacker's gains in information. In that case, the leakage of the system is not particularly interesting to the attacker, who already has enough information (as an example, consider the extreme $\pi_s = 1$).

Let us first examine the worst possible case, involving the highest possible posterior vulnerability, derived from the capacity. Intuitively, the worst case will be when the prior is evenly divided between two sites: $s$ and the one that differs the most from $s$. 

\begin{table}[b] 
  \centering
    \caption{Accuracy of the WF attack, for each method, when an s-distinguishing adversary possesses the worst (for the defender) possible $\pi$}
  \begin{tabular}{@{}cccc@{}}
    \toprule
    Method  & Accuracy & Recall & F1 score \\ 
    \midrule
    \rowcolor{gray!30}
    SEB exact  & 0.78 & 0.8 & 0.8 \\ 
    \hline
    \rowcolor{gray!12}
    SEB approx.  & 0.83 & 0.78 & 0.82 \\  
    \hline
    No Defense   & 0.89 & 0.88 & 0.89 \\ 
    \hline
    Average & 0.83 & 0.8 & 0.83 \\ 
    \hline
    Copy & 0.89 & 0.89 & 0.9 \\ 
    \hline
    Pad & 0.95 & 0.99 & 0.95 \\ 
    \bottomrule
  \end{tabular}
  \label{table:accuracy_worst_case}
\end{table}

For this omnipotent adversary equipped with the worst (for the defender) prior, \Cref{table:accuracy_worst_case} shows his accuracy  \footnote{Note that the accuracy could have also been calculated directly from \Cref{table:unkown-prior-s_dist} by simply multiplying the capacity by $\pi_s = 0.5$, as the capacity captures the worst possible leakage.}. The \emph{No Defense} gives a $89\%$ accuracy while \emph{SEB exact} offers an $11\%$ decrease in accuracy. All other methods offer at best only about half of that, with the best possible alternative to be \emph{Average} with an accuracy of $83\%$. \emph{Pad} is again worse than \emph{No Defense}, offering a staggering $95\%$ accuracy to the attack.

While in this first experiment we allocated the remaining $50\%$ of the prior to the most different site to capture the worst possible case, it will also be interesting to explore other ways to distribute it.
We examine two cases in \Cref{fig:50_50_prior} where we also include as a baseline the optimal solution for a  particular prior (Proposition \ref{optimazation_prop2}) but recall that it cannot be used in this scenario by the defender who does not have access to $\pi$.

First, we do a \emph{1-on-1} comparison for each site $t \in C \setminus\{s\}$ by creating a prior  $\Uni^{s,t}$ (i.e. we split the prior evenly between $s$ and $t$). \Cref{{fig:accuracy_by_duel}} shows that \emph{SEB exact} offers worse accuracy for the attacker compared to the \emph{Average} on every such 1-on-1 comparison. \emph{Copy} is close to \emph{No Defense} while \emph{Pad} is again the worst option in every case; the classifier can easily distinguish a site that produces page sizes that are multiples of 5KB.

It will also be interesting to keep $\pi_s = 0.5$ and split the remaining $50\%$ to some other $n$ sites, not uniformly but according to their traffic, making a \emph{1-on-n} comparison. 
\Cref{fig:accuracy_by_visits} shows that \emph{SEB exact} performs better in the worst case. Note that when $n=13$, the site with the most visits comes into play, notably affecting the prior \footnote{That is because the entire prior is normalized each time to ensure that its sum is 1.}.

On the other hand, \Cref{fig:accuracy_by_visits}  shows that \emph{Average} performs slightly better than \emph{SEB exact} when $n$ is increased.
To understand this, recall that in this scenario the optimal choice would  be a Weighted Average based on the specific prior. But \emph{Average} essentially assigns the same weight to each row of $C$, regardless of $n$, as it is prior-agnostic. However, as $n$ increases, more rows have a non-zero prior, favoring \emph{Average} since its uniform weights happen to work well with this particular prior. Also, keep in mind that the same heuristics were used to sample each site as discussed in \Cref{exp:setup}. This gives another advantage to \emph{Average}; consider, informally, that it attempts to find the average of similar-looking items. In reality, the way a user explores a site might differ for each site, making it difficult to generalize the performance of \emph{Average}, in contrast to \emph{SEB exact} which we proved to be the best option for the worst possible prior.

\subsection{Comparison of SEB exact and  SEB approx.}

\begin{figure*}[h]
\centering
\begin{subfigure}{.42\textwidth}
  \centering
		\includegraphics[width=\columnwidth]{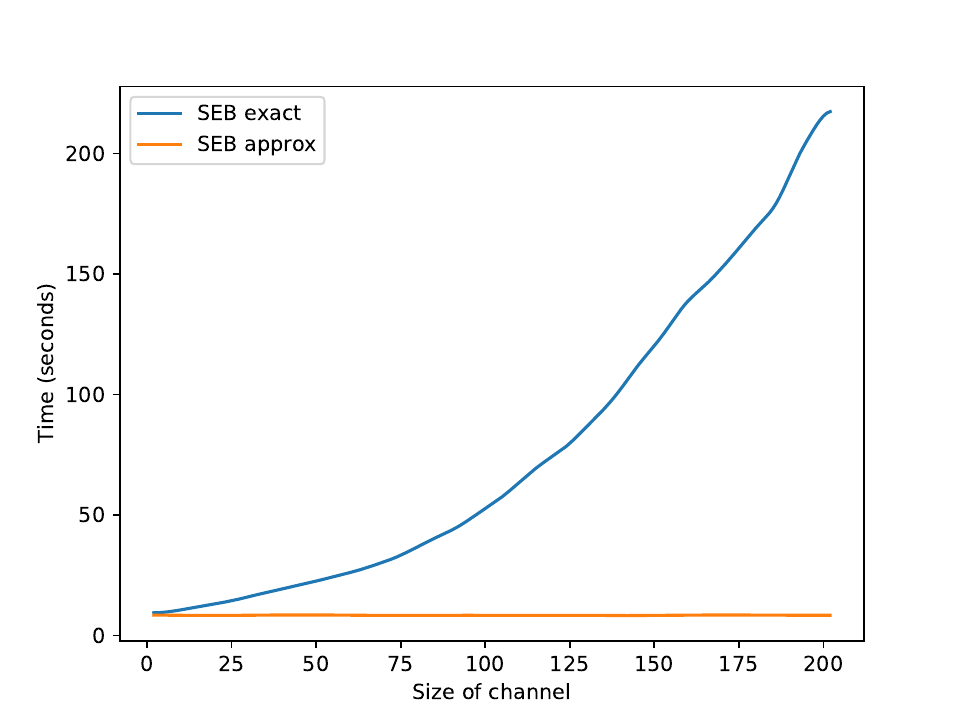}%
		\caption{Runtime}
		\label{fig:seb_timing}
\end{subfigure}%
\begin{subfigure}{.42\textwidth}
		\centering
		\includegraphics[width=\columnwidth]{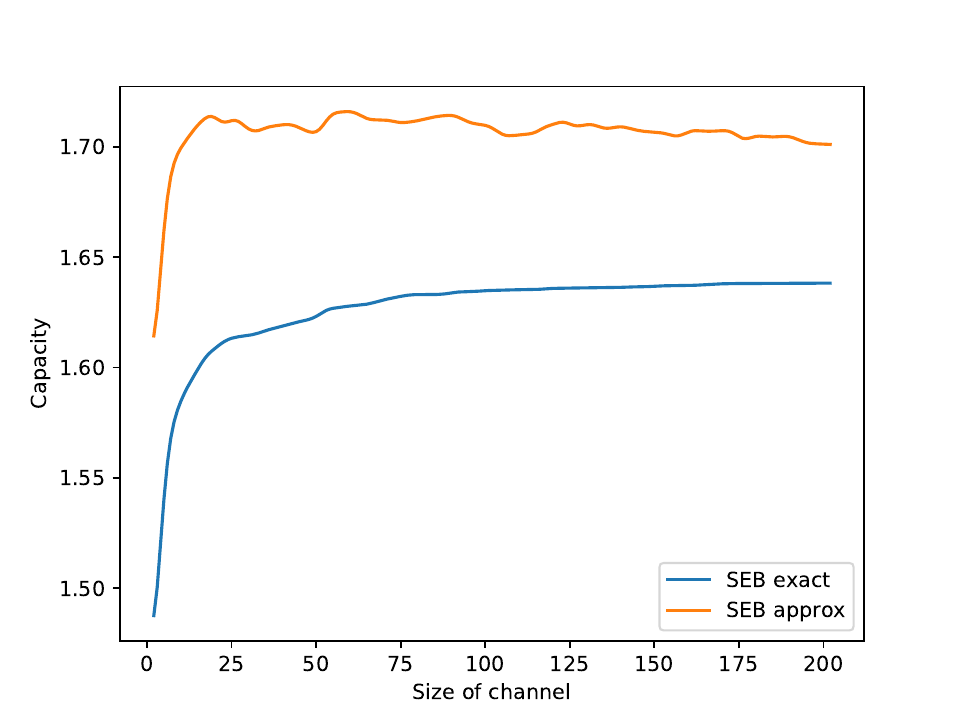}%
		\caption{Capacity}
		\label{fig:seb_capacity}
\end{subfigure}
\caption{Comparison of \emph{SEB exact} and \emph{SEB approx.}} 
\label{fig:seb_approx_vs_exact}
\end{figure*}

Previously, we discussed the trade-off between the two methods: the first offers optimal results at a high computational cost, while the latter provides approximated results quickly. 

To illustrate the comparison, we conduct an experiment to measure the performance and runtime of the two methods as the size of the channel increases \footnote{The system specifications used in the experiments can be found in \Cref{appendix:sys_spec}.}.
In this experiment, we employ all 200 sampled sites, instead of only the 20 sites used in the previous experiments. 

Figure \ref{fig:seb_timing} shows that \emph{SEB approx.} always completes almost instantly. On the contrary, Figure \ref{fig:seb_capacity} shows that as the channel size increases, \emph{SEB approx.} achieves a capacity near 1.7, compared to 1.6 for \emph{SEB exact}. On the other hand, the runtime of \emph{SEB exact} scales polynomially, even taking more than 3 minutes to complete when all the 200 sites are featured.

For smaller channel sizes, such as the one used in the previous experiments, choosing \emph{SEB exact} seems rather obvious; the runtime delay is insignificant while the improvement in capacity is remarkable.

\subsection{Scalability of Proposed Solutions}

\begin{figure*}[t]
\centering
\begin{subfigure}{.44 \textwidth}
  \centering
  \includegraphics[width=\linewidth]{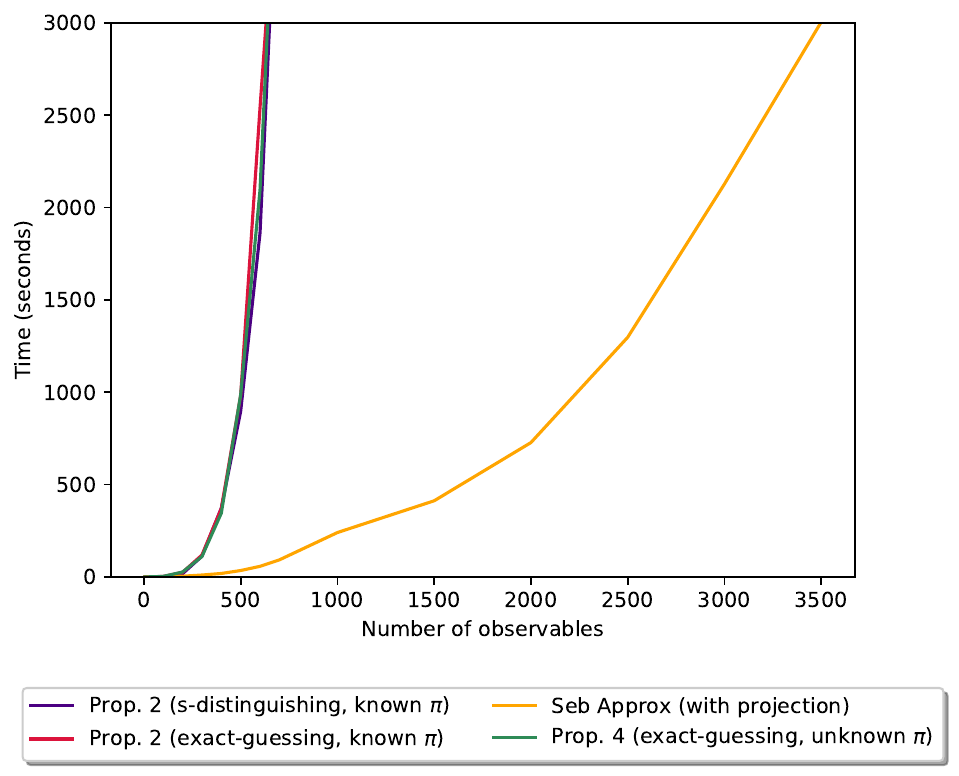}
  \caption{Prop. 2, Prop. 4 and SEB approx.}

  \label{fig:timing_rest}
\end{subfigure}%
\begin{subfigure}{.42 \textwidth}
  \centering
  \includegraphics[width=\linewidth]{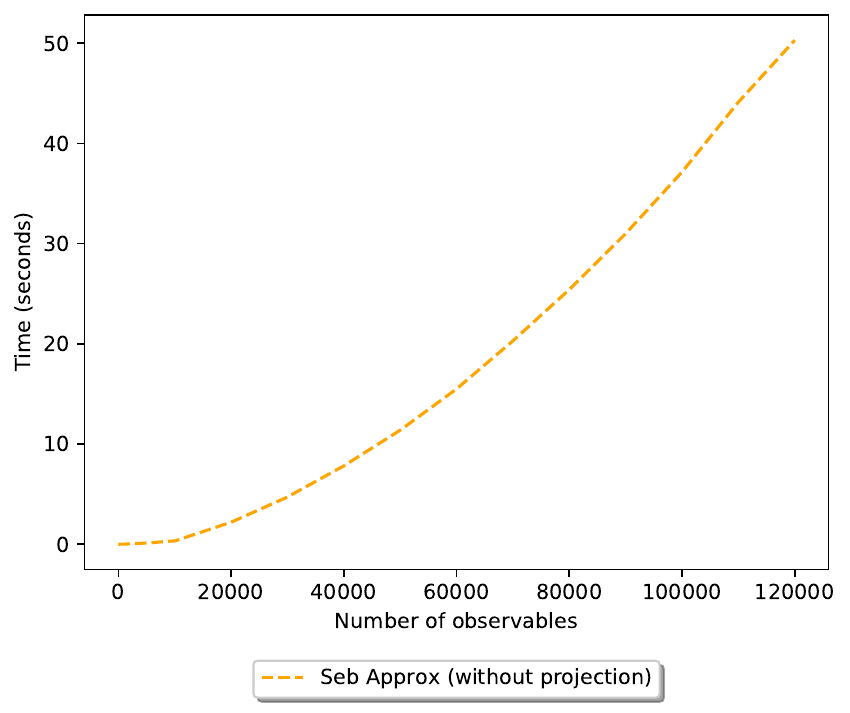}
  \caption{SEB approx.  without computing the projection to $\feasible$ }
  \label{fig:timing_seb_no_proj}
\end{subfigure}
\caption{Computation time of the proposed solutions as a function of observable attributes} 
\label{fig:50_50_prior}
\end{figure*}

After showing the better scalability of \emph{SEB approx.} (compared to \emph{SEB exact}), the next natural question is how the other proposed solutions scale. Indeed, this might be a primary concern for the defender if they wish to apply these solutions in a scenario with more observables.

\Cref{fig:timing_rest} shows the computation time of the corresponding LPs of \Cref{optimazation_prop2}, \Cref{prop:unknown-exact-lp}, and \emph{SEB approx.}; the latter is scaling much better than the first two, which behave similarly.

Recall, however, that \emph{SEB approx.}  consists of two steps: a) finding the approximated solution and b) projecting it into $\feasible$ via LP. Intuitively, the second step is more costly, which raises the question about the performance of \emph{SEB approx.}  in the unconstrained case. \Cref{fig:timing_seb_no_proj} shows that not computing the projection makes \emph{SEB approx.}  scale even better, as it computes a solution in less than a minute for even $120.000$ observables.

This can be useful to the defender if they are completely free to design their own row (e.g. creating a website from scratch). Moreover, if the defender can sacrifice some level of protection, they can boost their performance by computing the unconstrained case (i.e. without performing the costly projection) and then produce a result by sampling iteratively from the distribution until the constraints are met. For instance, if the defender has a page that is 100KB, they can iteratively sample until they get a number larger than 100, since everything less than 100 should be discarded (assuming that the page cannot be further compressed).

%% file: discussion.tex
\section{DISCUSSION AND CONCLUSIONS}

In this work, we explored methods for adding a new row to an existing information channel under two distinct types of adversaries for both known and unknown priors. 

When the prior is known, we discussed that the natural approach for an s-distinguishing adversary, namely an adversary that tries to distinguish if the secret is $s$ or not, falls short of achieving optimal results for the exact-guessing adversary (which tries to guess the exact secret in one try). In that case, LP can be used to provide an optimal solution.

We argued however that in real-world applications, the prior information available to an adversary is often difficult to estimate. Thus, when designing secure systems, one should also consider the worst-case scenario, or equivalently, in QIF terms, one should seek to minimize the capacity. Note that, when we use capacity for the comparison, the improvement is usually smaller than in the case of the leakage for a known prior. This is because the capacity represents the maximum leakage over all priors. Hence the capacities of the other approaches are ``squeezed'' between our capacity and the maximum possible capacity of a channel (system) with the same number of secrets and observables as the given one. 

Therefore, guided by minimizing capacity, we showed that for an unknown prior, any convex combination of the remaining rows is sufficient for the exact-guessing adversary. However, for the s-distinguishing adversary, solving the SEB (Smallest Enclosing Ball) problem is necessary to find a solution with optimal capacity, although it requires polynomial time.



Furthermore, we explained how our techniques can be applied to defend against website fingerprinting, specifically by discussing how a site can pad its responses to comply with our proposed solutions. We conducted experiments demonstrating that our approach can significantly reduce the leakage, compared to other natural methods. 
Then, we simulated an actual attack by training an ML classifier for the s-distinguishing adversary and an unknown prior. 

Our experiments confirm that solving the SEB problem ensures the lowest accuracy for the s-distinguishing adversary in the worst-case scenario compared to all the other prior-agnostic methods. Still the accuracy appears to be high because we considered an attacker who already knows that $s$ has a probability of $0.5$ to be the secret, in order to capture the capacity. This type of attacker is arguably uncommon in real-world situations. Although the abundance of information in today’s information age makes it nearly impossible to know the attacker’s prior knowledge, with our approach the defender can be prepared for the worst-case scenario.

We remark, however,  that the results presented in \Cref{sec:exp} depend on the assumptions taken during the design of the WF (Website Fingerprinting) attack simulation. The probability of each page at a given click-depth determines the behavior of each site. Informally, the more distinguishable the sites are, the harder it is to find a solution, and the benefit of the optimal solution, compared to other naive methods, may be reduced. Consider, for example, the extreme case where each website contains only a single page (e.g. appears with probability 1), which is unique in size. In this scenario, hiding s among a set of completely distinguishable sites becomes even more challenging.

Another practical aspect of the problem is the constraints faced when designing a defense against WF attacks. 
In this work, we considered only the page size, showing how linear constraints can be used. These constraints can similarly be applied to the packet size, which has been shown to be the most valuable information for a WF attack \cite{wb_features}.
However in real-world applications, one might want to include other parameters as well to increase the level of protection. Some natural choices are packet numbers, timing, and burst sizes, but an adversary can boost their WF attack by gathering information from other attributes, which can be as many as  $35683$ \cite{wb_features}.
This plethora of attributes raises the question of whether all of them can be expressed via linear constraints.
In cases where this is not possible, one potential approach would be to first find an unrestricted solution and then try to project it into the subspace of \lnorm{1} where the constraints are met. This approach was  discussed in the ''Feasible Solutions'' paragraph of 
\Cref{exp:feasible_solution} and we intend to explore it formally in future research. 


Future work should also continue by measuring the efficiency of our approaches in complex real-world  WF attacks, such as those studied in the literature for the Tor Network. Additional use cases could also be explored, as the scope of applications extends to any scenario in which a new user joins a fixed system and seeks privacy by designing their own responses based on the (fixed) responses of others. Finally, another potential research direction involves searching for a solution that is simultaneously optimal for both adversaries in the unknown prior setting while respecting any class of constraints.

%% file: 9_appendix.tex
%
\newenvironment{Reason}{\begin{tabbing}\hspace{2em}\= \hspace{1cm} \= \kill}
{\end{tabbing}\vspace{-1em}}
\newcommand\Step[2] {#1 \> $\begin{array}[t]{@{}llll}\displaystyle #2\end{array}$ \\}
\newcommand\StepR[3] {#1 \> $\begin{array}[t]{@{}llll}\displaystyle #3\end{array}$
\` {\RF \makebox[0pt][r]{\begin{tabular}[t]{r}``#2''\end{tabular}}} \\}
\newcommand\WideStepR[3] {#1 \>
$\begin{array}[t]{@{}ll}~\\\displaystyle #3\end{array}$ \`
{\RF \makebox[0pt][r]{\begin{tabular}[t]{r}``#2''\end{tabular}}} \\}
\newcommand\Space {~ \\}
\newcommand\RF {\small}

\subsection{Proofs of Section~\ref{sec:guessing-predicates}}

We start with two auxiliary lemmas.

\begin{lemma}\label{lem:maxdist-convex-hull}
	For any $S,T \subseteq \Reals^n$, it holds that
	\[
		\MaxDist{\ConvHull{S}}{\ConvHull{T}} \Wide= \MaxDist{S}{T}
		~,
	\]
	where distances are measured wrt any norm $\|\cdot\|$.
\end{lemma}
\begin{proof}
	Let $d = \MaxDist{S}{T}$. Since $S \subseteq \ConvHull{S}$ and $T \subseteq \ConvHull{T}$
	we clearly have
	$d \le \MaxDist{\ConvHull{S}}{\ConvHull{T}}$,
	the non-trivial part is to show that
	$d \ge \MaxDist{\ConvHull{S}}{\ConvHull{T}}$.

	We first show that
	\begin{equation}\label{eq4352}
	\forall s\in S, t \in \ConvHull{T} : \|s - t\| \le d
	~.
	\end{equation}
	Let $s \in S,t\in \ConvHull{T}$ and denote by $B_d[s]$ the closed
	ball of radius $d$ centered at $s$.
	Since $d \ge \|s'-t'\|$ for all $s'\in S,t'\in T$ it holds that
	\begin{align*}
		B_d[s]\kern4pt &\Wide\supseteq T ~, \\
	\intertext{and since balls induced by a norm are convex:}
		B_d[s] \Wide= \ConvHull{B_d[s]} &\Wide\supseteq \ConvHull{T}
		~,
	\end{align*}
	which implies $\|s - t\| \le d$, concluding the proof of \eqref{eq4352}.

	Finally we show that
	\[
		\forall s\in \ConvHull{S}, t \in \ConvHull{T} : \|s - t\| \le d
		~.
	\]
	Let $s\in \ConvHull{S}, t \in \ConvHull{T}$, from \eqref{eq4352} we know
	that $B_d[t] \supseteq S$, and since balls are convex we have that
	$B_d[t]  = \ConvHull{B_d[t]} \supseteq \ConvHull{S}$, which implies $\|s-t\| \le d$.
\end{proof}

Note that $\Diam{S} = \MaxDist{S}{S}$, hence \autoref{lem:maxdist-convex-hull} directly implies that
$
	\Diam{S} \Wide= \Diam{\ConvHull{S}}
$.

\begin{lemma}\label{lem:binary-bayes-capacity}
	Let $C : \secretspace\to\objectspace$ be a binary channel, with $\secretspace=\{s_1,s_2\}$.
	\begin{align*}
		\mulcapacity_\GEx(\Ch)
			&\Wide=
			\meleakage_\GEx(\Uni,\Ch)
			\Wide=
			1 + \frac{1}{2}\| C_{s_1} - C_{s_2} \|_1
			~, \\
		\mulcapacity_\LEx(\Ch)
			&\Wide=
			\meleakage_\LEx(\Uni,\Ch)
			\Wide=
			\frac{1}{1 - \frac{1}{2}\| C_{s_1} - C_{s_2} \|_1}
			~.
	\end{align*}
\end{lemma}
\begin{proof}
	The result for $\mulcapacity_\LEx(\Ch)$ follows directly from \autoref{thm:risk-mbc}.
	Define
	\[
		\textstyle
		\Sigma_\top \Wide= \sum_o\max_s C_{s,o}
		~,
		\quad
		\Sigma_\bot \Wide= \sum_o\min_s C_{s,o}
		~.
	\]
	We obtain the following equalities:
	\begin{align*}
		\textstyle\Sigma_\top \kern26pt           &\Wide= \mulcapacity_\GEx(C) ~,
			& \text{(\autoref{thm:mbc})} \\
		\textstyle\Sigma_\top + \Sigma_\bot &\Wide= 2 ~,
			& \text{(Rows sum to $1$)}\\
		\textstyle\Sigma_\top - \Sigma_\bot &\Wide= \| C_{s_1} - C_{s_2} \|_1 ~.
			& \text{(Def. of $\|\cdot\|_1$)}
	\end{align*}
	Adding the last two and substituting in the first gives the required result.
\end{proof}

\ThmPGuessingCapacity*
\begin{proof}
	Starting from  $\mulcapacity_{\GP}(C)$, let $\pi\in\distset\secretspace$,
	define $\rho^\pi,\WtoS$ as in \eqref{eq:rho},\eqref{eq:B} and let
	\[
		A \Wide= \WtoS C	~.
	\]
	Note that $A : \{ P, \neg P \} \to \objectspace$ is a binary channel;
	its rows $A_P$ and $A_{\neg P}$
	express the behavior of an ``average''
	(wrt $\pi$) secret of $C$ among those in $P$ and $\neg P$ respectively.
	
	Note that $P$ represents a \emph{set} of secrets of $C:\secretspace\to\objectspace$,
	but a \emph{single} secret of $A : \{ P, \neg P \} \to \objectspace$.
	Hence, with a slight abuse of notation, $C_P$ denotes a \emph{set of rows} of $C$,
	while $A_P$ denotes a \emph{single row} of $A$.
	
	We have that

	\begin{Reason}
	\Step{}{
		\meleakage_{\GP}(\pi,C)
	}
	\StepR{$=$}{\autoref{lem:p-vs-bayes}}{
		\meleakage_\GEx(\rho^\pi,A)
	}
	\StepR{$\le$}{Def. of $\mulcapacity_\GEx$}{
		\mulcapacity_\GEx(A)
	}
	\StepR{$=$}{\autoref{lem:binary-bayes-capacity}, $A$ is binary}{
		1 + \frac{1}{2} \| A_P  -  A_{\neg P} \|_1
	}
	\end{Reason}
	Notice that since $\WtoS$ is a channel, the rows of $A$ are \emph{convex combinations} of those
	of $C$. More precisely, since $\WtoS_{w,s} = 0$ whenever $s\not\in w$, the rows
	$A_P$ and $A_{\neg P}$ are convex combinations of the sets of rows
	$C_P$ and $C_{\neg P}$ respectively.
	Continuing the previous equational reasoning:
	\begin{Reason}
	\Step{}{
		1 + \frac{1}{2} \| A_P  -  A_{\neg P} \|_1
	}
	\WideStepR{$\le$}{$A_P \in \ConvHull{C_P}, A_{\neg P} \in \ConvHull{C_{\neg P}}$}{
		1 + \frac{1}{2} \MaxDist{\ConvHull{C_P}}{\ConvHull{C_{\neg P}}}
	}
	\StepR{$=$}{\autoref{lem:maxdist-convex-hull}}{
		1 + \frac{1}{2} \MaxDist{C_P}{C_{\neg P}}
	}
	\WideStepR{$=$}{Let $s\in P,t\in\neg P$ be those realizing $\MaxDist{\Ch_{P}}{\Ch_{\neg P}}$}{
		1 + \frac{1}{2} \| C_s - C_t \|_1 ~.
	}
	\end{Reason}
	This holds for all $\pi$, hence
	$\mulcapacity_{\GP}(C) \le 1 + \frac{1}{2} \| C_s - C_t \|_1$.
	Taking $\pi = \Uni^{s,t}$ we get $\rho^\pi = \Uni$,
	$A_P = C_s$ and $A_{\neg P} = C_t$, hence
	\[ \meleakage_{\GP}(\Uni^{s,t},C) = 1 + \frac{1}{2} \| C_s - C_t \|_1 ~. \]
	So the upper bound of $\mulcapacity_{\GP}(C)$ is attained, concluding the proof.

	For $\mulcapacity_{\LP}(C)$ the proof is similar.
\end{proof}

\subsection{Proofs of Section~\ref{sec:unknown-prior}}

\CHOptimalForExact*
\begin{proof}
	Starting from $\LEx$, note that $\Diam{S} = \MaxDist{S}{S}$. Hence \autoref{lem:maxdist-convex-hull}
	directly implies that
	$
		\Diam{S} \Wide= \Diam{\ConvHull{S}}
	$,
	that is
	taking convex combinations does not affect the diameter of a set.
	As a consequence
	$\Diam{\Ch^{\OptQ}} = \Diam{C_{\neg s}}$.
	Then
	$\mulcapacity_\LEx(\Ch^{\OptQ}) = \mulcapacity_\LEx(C_{\neg s})$
	follows from
	\autoref{thm:risk-mbc}.

	Similarly, for $\GEx$ the result
	follows from
	\autoref{thm:mbc} and the fact that convex combinations do not
	affect the column maxima.
\end{proof}

Recall the notation $S^\bot, S^\top$ from \eqref{eq9853}. We also denote by
$\preccurlyeq$ the partial order on $\Reals^n$ defied as
$x \preccurlyeq y$ iff $x_i \le y_i$ for all $i \in 1..n$.

\begin{lemma}\label{lem:seb-maxima}
	Let $S \subseteq \distset\objectspace$ and let
	\[
		\OptQ \in \argmin_{q \in \distset{\objectspace}} \MaxDist{q}{Q}
	\]
	be a solution
	to the $(\distset{\objectspace},\lnorm{1})$-SEB problem for $S$. Then 
	\[
		S^\bot  \Wide\preccurlyeq \OptQ  \Wide\preccurlyeq S^\top
		~.
	\]
\end{lemma}
\begin{proof}
	We show that $\OptQ  \Wide\preccurlyeq S^\top$, the proof of
	$S^\bot  \Wide\preccurlyeq \OptQ$ is similar.
	Assume that 
	$\OptQ[o_1] > S^\top_{o_1}$ for some $o_1$,
	that is $\OptQ[o_1]  > x_{o_1}$ for all $x\in S$.
	Select some
	\[
		0 \Wide< \epsilon
		\Wide<
		\min_{\substack{o\in\objectspace,x\in S \\ \OptQ[o]  \neq x_{o}}} |\OptQ[o] - x_{o}|
	\]
	and define
	$q \in \distset{\objectspace}$ as
	\[
		q_o
		\Wide=
		\begin{cases}
			\OptQ[o] - \epsilon
				& o = o_1 \\
			\OptQ[o] + \frac{\epsilon}{|\objectspace|-1}
				& \text{otherwise}
		\end{cases}
		~.
	\]
	In the following, we show that this construction moves $\OptQ$ \emph{strictly closer} to
	\emph{all} elements $x\in S$ simultaneously.

	Fix some arbitrary $x\in S$; the choice of $\epsilon$ is such that the
	relative order of $\OptQ[o]$ and $x_o$ is not affected by adding or subtracting $\epsilon$.
	Since $\OptQ[o_1] >  x_{o_1}$
	it follows that $ q_{o_1} > x_{o_1}$ which it turn implies that
	in the $o_1$ component we moved $\OptQ$ closer to $x$
	by exactly $\epsilon$:
	\begin{align}
		|q_{o_1} - x_{o_1}| &\Wide= |\OptQ[o_1] - x_{o_1}| - \epsilon
			~.\label{eq2622}\\
	\intertext{
		From the triangle inequality we get that
		in all other components, we moved $\OptQ$ away from $x$ 
		by at most $\frac{\epsilon}{|\objectspace|-1}$:
	}
		|q_o - x_o| &\Wide\le |\OptQ[o] - x_o| + \textstyle\frac{\epsilon}{|\objectspace|-1}
			~,
			& \forall o &\neq o_1 ~.
		\label{eq1486}\\
	\intertext{
		Moreover, since both vectors sum to 1 and
		$\OptQ[o_1] > x_{o_1}$ there must by some $o_2$ such that
		$\OptQ[o_2] < x_{o_2}$.
		By the choice of $\epsilon$ we get that
		$q_{o_2} < x_{o_2}$, hence:
	}
		|q_{o_2} - x_{o_2}| &\Wide< |\OptQ[o_2] - x_{o_2}| + \textstyle\frac{\epsilon}{|\objectspace|-1}
			~.\label{eq7553}
	\end{align}
	Summing over all $o$ using \eqref{eq2622}, \eqref{eq1486} and \eqref{eq7553} we get that
	\[
		\|q - x \|_1 \Wide< \|\OptQ - x\|_1 ~.
	\]
	Since this happens for all $x\in S$, it contradicts the fact that $\OptQ$ is a solution
	to the SEB problem, concluding the proof.
\end{proof}

\Minball*
\begin{proof}
	The minimization of $\mulcapacity_{\GS},\mulcapacity_{\LS}$ is a direct
	consequence of 	\autoref{thm:p-guessing-capacity}, since both capacities
	are increasing functions of $\MaxDist{C^q_s}{C^q_{\neg s}} = \MaxDist{q}{C_{\neg s}}$.
	
	For $\mulcapacity_{\LEx}$, let $t \neq s$. Since $q = \OptQ$ is a better
	choice than $q = C_t$
	(and $q = C_t$ is feasible since $\feasible = \distset\objectspace$)
	we have that:
	\begin{equation}\label{eq3623}
		\MaxDist{\OptQ}{C_{\neg s}}
		\Wide\le \MaxDist{C_t}{C_{\neg s}}
		\Wide\le \Diam{C_{\neg s}}
		~.
	\end{equation}
	As a consequence, for any $q \in \distset{\objectspace}$ it holds that
	\begin{Reason}
		\Step{}{
			\Diam{\channel^{\OptQ}_\secretspace}
		}
		\Step{$=$}{
			\max\{ \Diam{C_{\neg s}}, \MaxDist{\OptQ}{C_{\neg s}} \}
		}
		\Step{$=$}{
			\Diam{C_{\neg s}}
		}
		\Step{$\le$}{
			\Diam{\Ch^q_\secretspace}
			~.
		}
	\end{Reason}
	The result follows from \autoref{thm:risk-mbc}, since 
	$\mulcapacity_{\LEx}$ is an increasing function of $\Diam{C_\secretspace}$.

	The last case $\mulcapacity_{\GEx}$. From
	\autoref{lem:seb-maxima} 
	(which is applicable since $\feasible = \distset\objectspace$)
	we get that the elements of $\OptQ$ cannot
	be greater than the column maxima of $C_{\neg s}$.
	As a consequence,
	$C^{\OptQ}$ and $C_{\neg s}$ have exactly the same column maxima,
	which from \autoref{thm:mbc} implies that for any $q$:
	\[
		\mulcapacity_{\GEx}(C^{\OptQ})
		\Wide= \mulcapacity_{\GEx}(C_{\neg s})
 		\Wide\le \mulcapacity_{\GEx}(C^q)
		~.
	\]
	The last inequality comes from the fact that adding a row can only increase
	the column maxima.
\end{proof}

\subsection{System Specifications used in the experiments} \label{appendix:sys_spec}
The system specifications are crucial for determining the runtime of \emph{SEB exact} and \emph{SEB approx.} and the scalability of the proposed solutions. 

We implemented the experiments in Python 3.6.9 using the qif library and run them in the following system: 

\begin{itemize}
    \item CPU: 2 Quad core Intel Xeon E5-2623
    \item GPU:  NVIDIA GP102 GeForce GTX 1080 Ti
    \item RAM: 16x 16GB DDR4
\end{itemize}

\subsection{Selected Sites} \label{appendix:selected_sites}
In our experiments, out of the 200 sites we selected the 20 closest 
(in total variation distance) to the one we defend. The list is shown
\Cref{table:sites}.

\begin{table}[h]
    \centering
     \caption{Monthly Visits (in millions) of the selected sites (source: similarweb.com)}
    \begin{tabular}{|l|l|}
        \hline
        \textbf{Site Name} & \textbf{Visits} \\
        \hline
          \rowcolor{gray!30}
        mediafax.ro  ($s$) &  1.145 \\
        wort.lu & 0.177 \\
        sapo.pt & 13.6 \\
        sabah.com.tr & 61.4 \\
        lavanguardia.com & 41.5 \\
        e24.no & 0.8 \\
        cbc.na & 17.4 \\
        news.com.au & 8.3 \\
        cnnbrasil.com.br & 32.3 \\
        slobodnadalmacija.hr & 1.3 \\
        primicia.com.ve & 7.6 \\
        sme.sk & 7.4 \\
        topky.sk & 4.8 \\
        oe24.at & 1.2 \\
        ilfattoquotidiano.it & 12.7 \\
        meinbezirk.at & 2.6 \\
        voxeurop.eu & 0.019 \\
        the-european-times.com & 0.004 \\
        index.hu & 9.09 \\
        hespress.com & 5.4 \\
        \hline
    \end{tabular}
   
    \label{table:sites}
\end{table}